\documentclass[11pt]{article}
\usepackage[letterpaper,margin=1in]{geometry}
\usepackage{amsmath,amssymb,amsthm,bm}

\usepackage{thmtools}
\usepackage{thm-restate}

\usepackage{thmtools}

\declaretheorem[
	name=Theorem,
	numberwithin=section
	]{thm}
\declaretheorem[
	name=Lemma,
	sibling=thm,
	]{lem}
\declaretheorem[
	name=Proposition,
	sibling=thm,
	]{prop}
\declaretheorem[
	name=Corollary,
	sibling=thm,
	]{cor}

\declaretheorem[
	name=Definition,
	style=definition,
    sibling=thm
	]{defin}

\usepackage[english]{babel}
\usepackage{csquotes}

\usepackage{tabularx}

\usepackage{mathtools}
\usepackage{tikz-cd}
\usepackage{tikz}
\usetikzlibrary{shapes.geometric}
\usetikzlibrary{snakes}
\usepackage{thmtools}
\usepackage{wrapfig}
\usepackage[justification=centering]{caption}
\usepackage{xcolor} 
\usepackage{colortbl}
\numberwithin{equation}{section}
\usepackage[pdftex, bookmarks=true]{hyperref}
\usetikzlibrary{decorations.text}

\newenvironment{definition}{\begin{defin}}{\end{defin}}
\usepackage[T1]{fontenc}
\usepackage[utf8]{inputenc}

\newcommand{\summing}[2]{#1^{+}(#2)}
\newcommand{\vc}[1]{\mathbf{#1}} 

\newcommand{\PCSP}{\mathrm{PCSP}}
\newcommand{\diag}[1]{\mathcal{#1}}
\newcommand{\inst}[1]{\mathcal{#1}}
\newcommand{\lc}[1]{\mathrm{lc}(\inst{#1})}   \newcommand{\csp}{\mathrm{csp}}
\newcommand{\partialnpower}[2]{\mathrm{lc}_{\str A}^{#1}(\inst{#2})}
\newcommand{\partialkpower}[1]{\mathrm{lc}^k_{\str A}(#1)}
\newcommand{\RelaxLC}[1]{\mathrm{RLC}(#1)}
\newcommand{\RelaxLCn}[1]{\mathrm{RLC}_S(#1)}
\newcommand{\kthlevel}[1]{\mathrm{lvl}_k \, #1}
\newcommand{\kthpower}[1]{\mathrm{pow}_k \, #1}

\newcommand{\nthpower}[2]{\mathrm{pow}_{#1} \, #2}

\newcommand{\secondpower}[1]{\mathrm{pow}_2 \, #1}
\newcommand{\Ztwo}{\mathcal{Z}_2}

\newcommand{\VkZp}{\mathcal{V}_k\text{-}\mathbb{Z}_p}
\newcommand{\VkZpParam}[2]{\mathcal{V}_{#1}\text{-}\mathbb{Z}_{#2}}

\newcommand{\VAR}[1]{\mathtt{#1}}
\newcommand{\varx}{\VAR{x}}
\newcommand{\vary}{\VAR{y}}
\newcommand{\varz}{\VAR{z}}
\newcommand{\varu}{\VAR{u}}

\newcommand{\minion}[1]{\mathcal{#1}}
\newcommand{\cat}[1]{\mathcal{#1}}
\newcommand{\proj}{\pi}

\newcommand{\multiform}[1]{\mathfrak{#1}}
\newcommand{\weight}{\multiform{w}}
\newcommand{\ren}{\multiform{r}}
\newcommand{\formF}{\multiform{F}}
\newcommand{\formG}{\multiform{G}}
\newcommand{\formq}{\multiform{M}}
\newcommand{\mapf}{\tilde{f}}
\newcommand{\mapv}{\tilde{\vc{v}}}
\newcommand{\mapalpha}{\tilde{\alpha}}
\newcommand{\barP}{\bar{P}}
\newcommand{\spacQ}{Q}
\newcommand{\spacW}{W}

\DeclareMathOperator{\cm}{\odot}

\newcommand{\str}[1]{\mathbb{#1}}

\DeclareMathOperator{\Pol}{Pol}

\DeclareMathOperator{\id}{id}

\DeclareMathOperator{\ar}{ar}
\DeclareMathOperator{\arc}{arc}
\newcommand{\gl}[1]{\langle #1 \rangle}

\title{Toward a Uniform Algorithm and Uniform Reduction for Constraint Problems}

\date{}

\begin{document}

\author{Libor Barto, Maximilian Hadek, Dmitriy Zhuk\footnote{
All three authors were supported by the European Unions ERC Synergy Grant 101071674, POCOCOP.
The last author is also funded by 
the Czech Science Foundation project 25-16324S.
Views and opinions expressed are however those of the authors only and do not necessarily reflect those of the European Union or the European Research Council Executive Agency. Neither the European Union nor the granting authority can be held responsible for them.
}
} 

\maketitle

\begin{abstract}
We develop a unified framework to characterize the power of higher-level algorithms for the constraint satisfaction problem (CSP), such as $k$-consistency, the Sherali–Adams LP hierarchy, and the affine IP hierarchy. As a result, solvability of a fixed-template CSP or, more generally, a  Promise CSP by a given level is shown to depend only on the polymorphism minion of the template. Similarly, we obtain a minion-theoretic description of $k$-consistency reductions between Promise CSPs.

We introduce a new hierarchy of SDP-like vector relaxations with vectors over $\mathbb Z_{p}$
in which orthogonality is 
imposed on $k$-tuples of vectors. 
Surprisingly, this relaxation turns out to be equivalent to the
$k$-th level of the AIP-$\mathbb{Z}_p$ relaxation. We show that it solves
the CSP of the dihedral group $\mathbf{D}_4$, the smallest CSP that
fools the singleton BLP+AIP algorithm. Using this vector representation,
we further show that the $p$-th level of the $\mathbb{Z}_p$ relaxation
solves linear equations modulo $p^2$.
\end{abstract}

\maketitle

\section{Introduction}
The constraint satisfaction problem (CSP), in its most basic form, is the computational problem of deciding whether there is an assignment of values from some finite domain to a given set of variables that satisfies given constraints, where each constraint restricts possible values of a tuple of variables. As an example, the 2SAT instance $(\varx \vee \vary) \wedge (\neg \vary \vee \varz)$ can be regarded as an instance of the CSP: the variables are $\varx$, $\vary$, $\varz$, their domain consists of 0 (false) and 1 (true), and the constraints are that $\varx\vary$ must be in the relation $\{01,10,11\}$ and that $\vary\varz \in \{00, 01, 11\}$.

Among the major theoretical achievements in this area are two algorithms due to Bulatov~\cite{BulatovDichotomyFOCS,BulatovDichotomyArxiv} and Zhuk~\cite{ZhukDichotomyFOCS,zhukDichotomyJACM}. These algorithms are efficient and powerful in the following sense. By fixing a finite domain and a set of allowed relations, we obtain a restricted CSP, called a \emph{finite-domain fixed-template CSP}, which may be computationally as hard as possible (NP-complete). Whenever this is not the case, the algorithms run in polynomial time and correctly decide every instance of the restricted CSP. However, the algorithms have two major drawbacks.
First, they are complicated, and their correctness proofs are even more so. Second, they are not uniform: they are not polynomial-time procedures in the general, unrestricted setting; indeed, their running time is exponential in the size of the domains. These drawbacks make it difficult to generalize the algorithms to other variants of constraint problems, such as finite-domain fixed-template Promise CSPs~\cite{austrin20172+varepsilon,brakensiek2021promise,bbko2021,krokhinPCSPinvitation}, 
which we discuss later.

There are currently four basic uniform polynomial-time algorithms for CSPs: (generalized) arc consistency (AC), the basic linear programming relaxation (BLP), the basic semidefinite programming relaxation (SDP), and the affine integer programming relaxation (AIP). These are indeed relaxations in the sense that they never refute a Yes-instance (although they may fail to detect a No-instance).
Two natural ideas have proved useful for increasing the power of these basic algorithms.
In the singleton version of an algorithm, we take a variable $x$ and a value $a$ in its domain, run the algorithm with the additional constraint $x=a$, and, if it fails, remove $a$ from the domain of $x$. This process is repeated until some domain becomes empty (in which case we may safely refute) or no domain can be further pruned in this manner.
The second idea is to use higher levels: in the $k$-th level, we first construct an instance whose variables are all subsets of the original variables of size at most $k$, compute constraints implied by the original instance (in a certain polynomial-time fashion), and then run the original algorithm on this new instance.

The basic algorithms and these ideas can be used, tweaked, and/or combined in the hope that the resulting uniform algorithm is powerful enough for the problem class of interest. 
Unfortunately, this is not yet known to be the case even for the finite-domain fixed-template setting discussed above.
In fact, many such proposed algorithms are known to fail \cite{lichter2024limitations}. 
On the positive side, a relatively simple algorithm -- one that, in particular, does not use higher levels -- turned out to be quite powerful: the singleton version of  BLP+AIP, the combination of BLP and AIP from \cite{BraGurWroZiv-AIP+BLP}, denoted Singl(BLP+AIP). Using techniques from \cite{zhuk2021strong,zhuk2024simplified}, one can prove that the algorithm is correct on all instances of polynomial-time solvable fixed-template CSPs  with domain sizes up to 7~\cite{zhuk2025singAlgorithms}.

Is Singl(BLP+AIP) the uniform algorithm for fixed-template CSPs?
Unfortunately, a CSP suggested by Michael Kompatscher at the workshop CWC'2024  fools this algorithm. The instances of this restricted CSP are as follows. The domain of every variable is the 8-element dihedral group $\mathbf{D}_4$ (symmetries of a square) and constraints are of the form $x_1x_2\dots x_n \in R$, where $R$ is a coset of a subgroup of $\mathbf{D}_4^n$ (here $n \leq 4$ is sufficient to fool the algorithm~\cite{zhuk2025singAlgorithms}). 
Recall that the algorithm works up to domain sizes 7, so this example is tight. 

Having this concrete example at hand greatly helped to focus the research: it clarified what a uniform and powerful algorithm would need to overcome.
One aim was to show that no reasonable combinations of the basic algorithms together with the two ideas (singletons, higher levels) can ``solve $\mathbf{D}_4$''; this would demonstrate that some essentially new algorithmic idea is necessary. Perhaps the least understood aspect was the power of higher levels of AIP, largely because integer relaxations entered the area relatively recently, unlike the other basic algorithms. A second aim was to discover a new natural uniform algorithm capable of solving  $\mathbf{D}_4$.

Then in 2025, three results emerged in quick succession. First, we identified a simple relaxation, related to but distinct from the previously mentioned ones, that solves $\mathbf{D}_4$; moreover, its correctness proof is very short. Second, using a more technically involved argument, we proved that the second level of AIP, in fact, also solves $\mathbf{D}_4$. Third, we discovered a close connection between these two approaches.
The goal of this paper is to present some of the insights (and results) arising from this work.

\subsection{Relaxations}

The four basic relaxations can be conveniently defined using an alternative viewpoint on a CSP instance via the corresponding Label Cover (LC) instance. An \emph{LC instance} is a CSP instance in which variables have (possibly different) finite domains and every constraint is of the form $x=\alpha(y)$, where $\alpha$ is a function from the domain of $y$ to the domain of $x$. 
Given an instance $\inst{I}$ of a finite-domain CSP, we associate to it an LC instance $\diag{D} = \lc{\inst{I}}$  as follows. The variables of $\diag{D}$ are the variables of $\inst{I}$, with the same domains as in $\inst{I}$, together with one additional variable for each constraint of $\inst{I}$. For a constraint $x_1x_2\dots x_n \in R$, the corresponding new variable $y$ has domain $R$ and we impose in $\diag{D}$ the constraints $x_1 = \pi_1(y)$, $x_2 = \pi_2(y)$, \dots $x_n = \pi_n(y)$, where $\pi_i$ is the projection function $a_1a_2\dots a_n \mapsto a_i$.

As an example, consider the CSP instance 
$\varx=1, \varx\neq \vary$, $\vary\le \varz$ over the domain $\{0,1\}$.
The associated LC instance is shown in Figure~\ref{fig:2SatDiagram}. 
Figure \ref{fig:2SatDetailedDiagram} shows the same LC instance in greater detail.

\begin{figure}[h]
    \centering
    \begin{minipage}{0.45\textwidth}
        \centering
\begin{tikzpicture}[
    vertex/.style={circle, fill, inner sep=2pt},
    lbl/.style={sloped, above, midway, font=\scriptsize, inner sep=1pt}
]
\node[vertex, label=left:{$\varx=1$}]  (Rx1) at (0,  1.5) {};
\node[vertex, label=left:{$\varx \neq \vary$}](Rxy) at (0,  0.5) {};
\node[vertex, label=left:{$\vary\le \varz$}](Ryz) at (0, -0.5) {};
\node[vertex, label=right:{$\varx$}] (x) at (4,  1) {};
\node[vertex, label=right:{$\vary$}] (y) at (4,  0) {};
\node[vertex, label=right:{$\varz$}] (z) at (4, -1) {};
\draw[-latex] (Rx1) to node[lbl]           {$\proj_{1}$} (x);
\draw[-latex] (Rxy) to node[lbl, pos=0.35] {$\proj_{1}$} (x);
\draw[-latex] (Rxy) to node[lbl]           {$\proj_{2}$} (y);
\draw[-latex] (Ryz) to node[lbl]           {$\proj_{1}$} (y);
\draw[-latex] (Ryz) to node[lbl, pos=0.65] {$\proj_{2}$} (z);
\end{tikzpicture}
        \caption{An example of an LC instance.}\label{fig:2SatDiagram}
    \end{minipage}
\hfill
\begin{minipage}{0.45\textwidth}
    \centering
\definecolor{darkblue}{RGB}{0,0,139}
\begin{tikzpicture}[
    tuple/.style={inner sep=1pt, font=\small},
    tupleB/.style={inner sep=1pt, font=\small\bfseries\boldmath, color=darkblue},
    label/.style={font=\small, inner sep=0pt, anchor=north east},
    rlabel/.style={font=\small, inner sep=0pt, anchor=north west},
    dot/.style={circle, fill, inner sep=1.5pt}
]
\begin{scope}
  \clip (0, 3.0) ellipse (0.55cm and 0.4cm);
  \fill[blue!8!white] (-1, 2.5) rectangle (1, 3.5);
\end{scope}
\draw (0, 3.0) ellipse (0.55cm and 0.4cm);
\node[tupleB] (Rx1t1) at (0, 3.0) {$1$};
\node[label] at (-0.35, 2.65) {$\varx=1$};
\begin{scope}
  \clip (0, 1.2) ellipse (0.65cm and 0.65cm);
  \fill[blue!8!white] (-1, 0.4) rectangle (1, 1.9);
\end{scope}
\draw (0, 1.2) ellipse (0.65cm and 0.65cm);
\node[tuple]  (Rxy01) at (0, 1.45) {$01$};
\node[tupleB] (Rxy10) at (0, 0.95) {$10$};
\node[label] at (-0.35, 0.6) {$\varx\neq \vary$};
\begin{scope}
  \clip (0, -1.2) ellipse (0.65cm and 0.9cm);
  \fill[blue!8!white] (-1, -2.2) rectangle (1, -0.2);
\end{scope}
\draw (0, -1.2) ellipse (0.65cm and 0.9cm);
\node[tupleB] (Ryz00) at (0, -0.8)  {$00$};
\node[tuple]  (Ryz01) at (0, -1.2)  {$01$};
\node[tuple]  (Ryz11) at (0, -1.6)  {$11$};
\node[label] at (-0.35, -1.95) {$\vary\le \varz$};
\begin{scope}
  \clip (4, 2.6) ellipse (0.5cm and 0.7cm);
  \fill[blue!8!white] (3, 1.8) rectangle (5, 3.4);
\end{scope}
\draw (4, 2.6) ellipse (0.5cm and 0.7cm);
\node[tuple]  (x0) at (4, 2.9) [xshift=4pt] {$0$};
\node[tupleB] (x1) at (4, 2.3) [xshift=4pt] {$1$};
\node[rlabel] at (4.4, 1.95) {$\varx$};
\begin{scope}
  \clip (4, 0.6) ellipse (0.5cm and 0.7cm);
  \fill[blue!8!white] (3, -0.2) rectangle (5, 1.4);
\end{scope}
\draw (4, 0.6) ellipse (0.5cm and 0.7cm);
\node[tupleB] (y0) at (4, 0.9) [xshift=4pt] {$0$};
\node[tuple]  (y1) at (4, 0.3) [xshift=4pt] {$1$};
\node[rlabel] at (4.4, -0.05) {$\vary$};
\begin{scope}
  \clip (4, -1.4) ellipse (0.5cm and 0.7cm);
  \fill[blue!8!white] (3, -2.2) rectangle (5, -0.6);
\end{scope}
\draw (4, -1.4) ellipse (0.5cm and 0.7cm);
\node[tupleB] (z0) at (4, -1.1) [xshift=4pt] {$0$};
\node[tuple]  (z1) at (4, -1.7) [xshift=4pt] {$1$};
\node[rlabel] at (4.4, -2.05) {$\varz$};
\draw[-latex, very thick, darkblue, shorten <=0.1cm, shorten >=0.15cm] ([xshift=3pt]Rx1t1.east) to ([xshift=-3pt]x1.west);
\draw[-latex, shorten <=0.1cm, shorten >=0.15cm]                       ([xshift=3pt]Rxy01.east) to ([xshift=-3pt]x0.west);
\draw[-latex, shorten <=0.1cm, shorten >=0.15cm]                       ([xshift=3pt]Rxy01.east) to ([xshift=-3pt]y1.west);
\draw[-latex, very thick, darkblue, shorten <=0.1cm, shorten >=0.15cm] ([xshift=3pt]Rxy10.east) to ([xshift=-3pt]x1.west);
\draw[-latex, very thick, darkblue, shorten <=0.1cm, shorten >=0.15cm] ([xshift=3pt]Rxy10.east) to ([xshift=-3pt]y0.west);
\draw[-latex, very thick, darkblue, shorten <=0.1cm, shorten >=0.15cm] ([xshift=3pt]Ryz00.east) to ([xshift=-3pt]y0.west);
\draw[-latex, very thick, darkblue, shorten <=0.1cm, shorten >=0.15cm] ([xshift=3pt]Ryz00.east) to ([xshift=-3pt]z0.west);
\draw[-latex, shorten <=0.1cm, shorten >=0.15cm]                       ([xshift=3pt]Ryz01.east) to ([xshift=-3pt]y0.west);
\draw[-latex, shorten <=0.1cm, shorten >=0.15cm]                       ([xshift=3pt]Ryz01.east) to ([xshift=-3pt]z1.west);
\draw[-latex, shorten <=0.1cm, shorten >=0.15cm]                       ([xshift=3pt]Ryz11.east) to ([xshift=-3pt]y1.west);
\draw[-latex, shorten <=0.1cm, shorten >=0.15cm]                       ([xshift=3pt]Ryz11.east) to ([xshift=-3pt]z1.west);
\fill[blue] ([xshift=3pt]Rx1t1.east) circle (1.5pt);
\fill[blue] ([xshift=3pt]Rxy01.east) circle (1.5pt);
\fill[blue] ([xshift=3pt]Rxy10.east) circle (1.5pt);
\fill[blue] ([xshift=3pt]Ryz00.east) circle (1.5pt);
\fill[blue] ([xshift=3pt]Ryz01.east) circle (1.5pt);
\fill[blue] ([xshift=3pt]Ryz11.east) circle (1.5pt);
\fill[blue] ([xshift=-3pt]x0.west) circle (1.5pt);
\fill[blue] ([xshift=-3pt]x1.west) circle (1.5pt);
\fill[blue] ([xshift=-3pt]y0.west) circle (1.5pt);
\fill[blue] ([xshift=-3pt]y1.west) circle (1.5pt);
\fill[blue] ([xshift=-3pt]z0.west) circle (1.5pt);
\fill[blue] ([xshift=-3pt]z1.west) circle (1.5pt);
\end{tikzpicture}
        
        \caption{A detailed example of an LC instance and its solution.}\label{fig:2SatDetailedDiagram}
    \end{minipage}
\end{figure}

Notice that $\diag D$ and $\inst I$ are essentially the same instances: solutions to $\inst I$ are exactly solutions to $\diag D$ together with witnesses certifying that the constraints are satisfied. In the above example, the solution $\varx \mapsto 1$,
$\vary \mapsto 0$,
$\varz \mapsto 0$ to $\inst I$ 
 corresponds to the solution
$\varx \mapsto 1$,
$\vary \mapsto 0$,
$\varz \mapsto 0$,
$(\varx=1)\mapsto 1$, 
$(\varx\neq \vary)\mapsto 10$, 
$(\vary\le \varz)\mapsto 00$
 to $\inst D$ (see Figure \ref{fig:2SatDetailedDiagram}).

\setlength{\columnsep}{0pt}

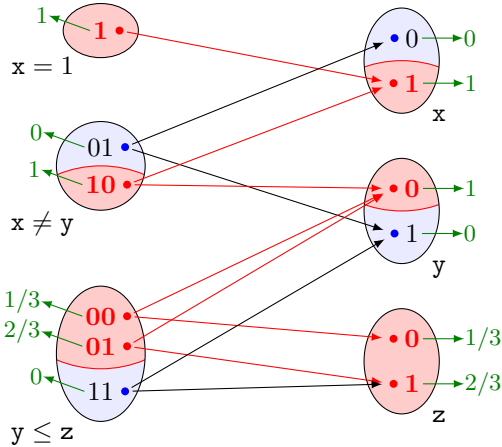
\begin{wrapfigure}{l}[6pt]{0.55\textwidth}
    \centering
\definecolor{darkblue}{RGB}{0,0,139}
\definecolor{dgreen}{RGB}{0,130,0}
\begin{tikzpicture}[
    tuple/.style={inner sep=1pt, font=\small},
    tupleB/.style={inner sep=1pt, font=\small\bfseries\boldmath, color=red},
    label/.style={font=\small, inner sep=0pt, anchor=north east},
    rlabel/.style={font=\small, inner sep=0pt, anchor=north west},
    blp/.style={color=dgreen, font=\footnotesize, inner sep=0pt}
]

\begin{scope}
  \clip (0, 3.0) ellipse (0.50cm and 0.36cm);
  \fill[red!20!white] (-1, 2.5) rectangle (1, 3.5);
\end{scope}
\draw (0, 3.0) ellipse (0.50cm and 0.36cm);
\node[tupleB] (Rx1t1) at (0, 3.0) {$1$};
\node[label] at (-0.35, 2.65) {$\varx=1$};

\begin{scope}
  \clip (0, 1.2) ellipse (0.59cm and 0.59cm);
  \fill[blue!8!white] (-1, 0.4) rectangle (1, 1.9);
\end{scope}
\begin{scope}
  \clip (0, 1.2) ellipse (0.59cm and 0.59cm);
  \clip (0, -0.3) circle (1.5cm);
  \fill[red!20!white] (-1, 0.4) rectangle (1, 1.9);
\end{scope}
\begin{scope}
  \clip (0, 1.2) ellipse (0.59cm and 0.59cm);
  \draw[red] (0, -0.3) circle (1.5cm);
\end{scope}
\draw (0, 1.2) ellipse (0.59cm and 0.59cm);
\node[tuple]  (Rxy01) at (0, 1.45) {$01$};
\node[tupleB] (Rxy10) at (0, 0.95) {$10$};
\node[label] at (-0.35, 0.6) {$\varx\neq \vary$};

\begin{scope}
  \clip (0, -1.3) ellipse (0.59cm and 0.90cm);
  \fill[blue!8!white] (-1, -2.4) rectangle (1, -0.2);
\end{scope}
\begin{scope}
  \clip (0, -1.3) ellipse (0.59cm and 0.90cm);
  \clip (0, 0.0) circle (1.5cm);
  \fill[red!20!white] (-1, -2.4) rectangle (1, -0.2);
\end{scope}
\begin{scope}
  \clip (0, -1.3) ellipse (0.59cm and 0.90cm);
  \draw[red] (0, 0.0) circle (1.5cm);
\end{scope}
\draw (0, -1.3) ellipse (0.59cm and 0.90cm);
\node[tupleB] (Ryz00) at (0, -0.8)  {$00$};
\node[tupleB] (Ryz01) at (0, -1.2)  {$01$};
\node[tuple]  (Ryz11) at (0, -1.8)  {$11$};
\node[label] at (-0.35, -2.2) {$\vary\le \varz$};

\begin{scope}
  \clip (4, 2.6) ellipse (0.5cm and 0.7cm);
  \fill[blue!8!white] (3, 1.8) rectangle (5, 3.4);
\end{scope}
\begin{scope}
  \clip (4, 2.6) ellipse (0.5cm and 0.7cm);
  \clip (4, 1.1) circle (1.5cm);
  \fill[red!20!white] (3, 1.8) rectangle (5, 3.4);
\end{scope}
\begin{scope}
  \clip (4, 2.6) ellipse (0.5cm and 0.7cm);
  \draw[red] (4, 1.1) circle (1.5cm);
\end{scope}
\draw (4, 2.6) ellipse (0.5cm and 0.7cm);
\node[tuple]  (x0) at (4, 2.9) [xshift=4pt] {$0$};
\node[tupleB] (x1) at (4, 2.3) [xshift=4pt] {$1$};
\node[rlabel] at (4.4, 1.95) {$\varx$};

\begin{scope}
  \clip (4, 0.6) ellipse (0.5cm and 0.7cm);
  \fill[blue!8!white] (3, -0.2) rectangle (5, 1.4);
\end{scope}
\begin{scope}
  \clip (4, 0.6) ellipse (0.5cm and 0.7cm);
  \clip (4, 2.1) circle (1.5cm);
  \fill[red!20!white] (3, -0.2) rectangle (5, 1.4);
\end{scope}
\begin{scope}
  \clip (4, 0.6) ellipse (0.5cm and 0.7cm);
  \draw[red] (4, 2.1) circle (1.5cm);
\end{scope}
\draw (4, 0.6) ellipse (0.5cm and 0.7cm);
\node[tupleB] (y0) at (4, 0.9) [xshift=4pt] {$0$};
\node[tuple]  (y1) at (4, 0.3) [xshift=4pt] {$1$};
\node[rlabel] at (4.4, -0.05) {$\vary$};

\begin{scope}
  \clip (4, -1.4) ellipse (0.5cm and 0.7cm);
  \fill[red!20!white] (3, -2.2) rectangle (5, -0.6);
\end{scope}
\draw (4, -1.4) ellipse (0.5cm and 0.7cm);
\node[tupleB] (z0) at (4, -1.1) [xshift=4pt] {$0$};
\node[tupleB] (z1) at (4, -1.7) [xshift=4pt] {$1$};
\node[rlabel] at (4.4, -2.05) {$\varz$};

\draw[-latex, red, shorten <=0.1cm, shorten >=0.15cm] ([xshift=3pt]Rx1t1.east) to ([xshift=-3pt]x1.west);
\draw[-latex, shorten <=0.1cm, shorten >=0.15cm]      ([xshift=3pt]Rxy01.east) to ([xshift=-3pt]x0.west);
\draw[-latex, shorten <=0.1cm, shorten >=0.15cm]      ([xshift=3pt]Rxy01.east) to ([xshift=-3pt]y1.west);
\draw[-latex, red, shorten <=0.1cm, shorten >=0.15cm] ([xshift=3pt]Rxy10.east) to ([xshift=-3pt]x1.west);
\draw[-latex, red, shorten <=0.1cm, shorten >=0.15cm] ([xshift=3pt]Rxy10.east) to ([xshift=-3pt]y0.west);
\draw[-latex, red, shorten <=0.1cm, shorten >=0.15cm] ([xshift=3pt]Ryz00.east) to ([xshift=-3pt]y0.west);
\draw[-latex, red, shorten <=0.1cm, shorten >=0.15cm] ([xshift=3pt]Ryz00.east) to ([xshift=-3pt]z0.west);
\draw[-latex, red, shorten <=0.1cm, shorten >=0.15cm] ([xshift=3pt]Ryz01.east) to ([xshift=-3pt]y0.west);
\draw[-latex, red, shorten <=0.1cm, shorten >=0.15cm] ([xshift=3pt]Ryz01.east) to ([xshift=-3pt]z1.west);
\draw[-latex, shorten <=0.1cm, shorten >=0.15cm]      ([xshift=3pt]Ryz11.east) to ([xshift=-3pt]y1.west);
\draw[-latex, shorten <=0.1cm, shorten >=0.15cm]      ([xshift=3pt]Ryz11.east) to ([xshift=-3pt]z1.west);

\fill[red]  ([xshift=3pt]Rx1t1.east) circle (1.5pt); 
\fill[blue] ([xshift=3pt]Rxy01.east) circle (1.5pt); 
\fill[red]  ([xshift=3pt]Rxy10.east) circle (1.5pt); 
\fill[red]  ([xshift=3pt]Ryz00.east) circle (1.5pt); 
\fill[red]  ([xshift=3pt]Ryz01.east) circle (1.5pt); 
\fill[blue] ([xshift=3pt]Ryz11.east) circle (1.5pt); 

\fill[blue] ([xshift=-3pt]x0.west) circle (1.5pt); 
\fill[red]  ([xshift=-3pt]x1.west) circle (1.5pt); 
\fill[red]  ([xshift=-3pt]y0.west) circle (1.5pt); 
\fill[blue] ([xshift=-3pt]y1.west) circle (1.5pt); 
\fill[red]  ([xshift=-3pt]z0.west) circle (1.5pt); 
\fill[red]  ([xshift=-3pt]z1.west) circle (1.5pt); 

\draw[-latex, dgreen] (Rx1t1.west) -- ++(-0.55, 0.2) node[left, blp] {$1$};
\draw[-latex, dgreen] (Rxy01.west) -- ++(-0.55, 0.2) node[left, blp] {$0$};
\draw[-latex, dgreen] (Rxy10.west) -- ++(-0.55, 0.2) node[left, blp] {$1$};
\draw[-latex, dgreen] (Ryz00.west) -- ++(-0.55, 0.2) node[left, blp] {$1/3$};
\draw[-latex, dgreen] (Ryz01.west) -- ++(-0.55, 0.2) node[left, blp] {$2/3$};
\draw[-latex, dgreen] (Ryz11.west) -- ++(-0.55, 0.2) node[left, blp] {$0$};

\draw[-latex, dgreen] (x0.east) -- ++(0.55, 0) node[right, blp] {$0$};
\draw[-latex, dgreen] (x1.east) -- ++(0.55, 0) node[right, blp] {$1$};
\draw[-latex, dgreen] (y0.east) -- ++(0.55, 0) node[right, blp] {$1$};
\draw[-latex, dgreen] (y1.east) -- ++(0.55, 0) node[right, blp] {$0$};
\draw[-latex, dgreen] (z0.east) -- ++(0.55, 0) node[right, blp] {$1/3$};
\draw[-latex, dgreen] (z1.east) -- ++(0.55, 0) node[right, blp] {$2/3$};
\end{tikzpicture}
    \caption{AC and BLP solutions.}\label{fig:2SatACBLPDiagram}
\end{wrapfigure}

The AC relaxation, informally, gradually removes inconsistent domain elements for variables in $\diag D
$ until the process stabilizes. The instance is refuted if some domain becomes empty. Otherwise, we obtain, for each variable $x$ in $\diag D$, a nonempty subset $\diag D'_x$ of its original domain $\diag D_x$ such that for every constraint $x = \alpha(y)$, the image $\alpha(\diag D'_y)$ is equal to $\diag D'_x$. Although this is only an informal description of AC, an explicit AC procedure can be reconstructed from this requirement. For the instance in Figure \ref{fig:2SatDetailedDiagram}, these subsets are shown in red in Figure \ref{fig:2SatACBLPDiagram}.

The BLP relaxation accepts if there exist functions $f_x\colon \diag D_x \to [0,1]$ (one function $f_x$ for each variable $x$ in $\diag D$) such that, for every variable $x$,  $\sum\limits_{a \in \diag D_x} f_x(a)=1$, and, for every constraint $x = \alpha(y)$ and every $a \in D_x$, $f_x(a) = \sum\limits_{b; \alpha(b) = a} f_x(b)$; we refer to the latter requirement as the \emph{summation condition}. 
The existence of such functions can be verified by a linear program. 
For the instance in Figure~\ref{fig:2SatDetailedDiagram}, these functions are shown in green in Figure~\ref{fig:2SatACBLPDiagram}.

Observe that the BLP relaxation is stronger than AC: given  functions $f_x$ satisfying the BLP conditions, the subsets $\diag D'_x = \{a \in \diag D_x; f_x(a)>0\}$ satisfy the AC conditions.

The SDP relaxation strengthens BLP further. It plays a major role in approximation variants of the CSP \cite{beginningofSDPGoemansWilliamson,raghavendra2008optimal,barto2016robustly,braGur2023sdp}. 
The following formulation for the decision setting was worked out in \cite{braGur2023sdp,ciardo2025semidefinite}. 
The SDP relaxation accepts if there exists $N \in \mathbb{N}$ and functions $f_x$ from $\diag{D}_x$ to the $N$-dimensional Euclidean space $\mathbb{R}^N$ such that $\sum_{a \in D_x} f_x(a)=\mathbf{i} $ is a unit-length vector independent of $x$, the vectors $f_x(a)$ are pairwise orthogonal for every $x$, and we have syntactically the same summation condition as for the BLP, but this time we sum in $\mathbb{R}^N$. 
Observe that SDP is indeed stronger than BLP, since the BLP solution  can be obtained from the SDP solution by taking norms.

Finally, AIP relaxation is formally similar to BLP, except that   the interval $[0,1]$ is replaced by $\mathbb{Z}$, the integers, so that now $f_x\colon \diag D_x \to \mathbb{Z}$. For our purposes, weaker variants are particularly relevant, namely the $\mathbb{Z}_p$ relaxations for primes $p$, obtained by replacing $\mathbb{Z}$ with $\mathbb Z_p$, the integers modulo $p$.
Although these relaxations resemble BLP and SDP syntactically,
their algorithmic strength is incomparable. SDP (which strengthens BLP) correctly solves all finite‑domain fixed‑template CSPs that are solvable by consistency methods \cite{barto2016robustly,braGur2023sdp}, whereas no level of AIP is capable of doing so. On the other hand, the $\mathbb{Z}_p$ relaxation correctly decides systems of linear equations over $\mathbb{Z}_p$, but no level of SDP suffices \cite{schoenebeck2008linear,tulsiani2009csp}.

\subsection{Minion descriptions}

The four relaxations can be described in a unified way. Fix a construction $\minion{R}$ that assigns to every finite set $D$ a set $\minion{R}^{(D)}$, and to every function $\alpha\colon D \to E$ a function $\minion{R}^{(\alpha)}\colon \minion{R}^{(D)} \to \minion{R}^{(E)}$. Intuitively, $\minion{R}^{(D)}$ is a relaxed version of the domain $D$ and $\minion{R}^{(\alpha)}$ a relaxed version of the function $\alpha$. Given such a construction $\minion{R}$ and an LC instance $\diag D$, the \emph{$\minion{R}$-relaxation} of $\diag D$ is obtained by replacing each domain $\diag D_x$ with $\minion{R}^{(\diag D_x)}$ and each constraint $x = \alpha(y)$ with $x = \minion{R}^{(\alpha)} \, (y)$. 
If the construction $\minion{R}$ is functorial (preserves the composition of functions and identities) and nontrivial (assigns a nonempty set to a nonempty set), then $\minion{R}$ is indeed a relaxation in the sense that the $\minion{R}$-relaxation of every solvable LC instance $\inst{D}$ is itself solvable \cite{adamekGummTrnkova}.  Such constructions are called (abstract) \emph{minions}.

The four relaxations are \emph{described} by a minion $\minion{R}$ in the sense that they accept an LC instance $\diag D$ if and only if the $\minion{R}$-relaxation of $\diag D$ has a solution. For AC,  $\minion{R}^{(D)}$ is the set of all nonempty subsets of $D$, and $\minion{R}^{(\alpha)}$ maps a subset to its $\alpha$-image. For BLP, $\minion{R}^{(D)}$ is the set of functions $f\colon D \to [0,1]$ satisfying $\sum_{a \in D} f(a) = 1$, and $\minion{R}^{(\alpha)}$ is defined so as to enforce the summation condition. The corresponding  minions $\minion{R}$ for SDP, AIP and $\mathbb{Z}_p$ relaxations are defined in an entirely analogous manner. Minions $\minion{R}$ and $\minion{S}$ can be compared using a \emph{minion homomorphism}: a family of functions $\xi^{(D)}\colon \minion{R}^{(D)} \to \minion{S}^{(D)}$ that commute with every $\minion{\cdot}^{(\alpha)}$. For example, the transition from BLP to AC described above is a minion homomorphism.

Having a description of a relaxation by a minion is particularly useful in the context of Promise CSPs (as discussed in the next subsection). For this reason, researchers in the area try to obtain such descriptions for specific combinations, singleton versions, and higher levels of the basic relaxations \cite{ciardo2023clap,ciardoZivny2023hierarchies,zhuk2025singAlgorithms}.

This is often possible for singleton versions, using an idea originating from \cite{ciardo2023clap} to record multiple runs of an algorithm as sequences of elements of $\minion{R}^{(D)}$. (We do not provide details in this paper.) 
For higher levels, however, it was unclear how a single relaxed domain could encode information concerning more variables at once. Indeed, a common intuition until recently was that this is impossible except in specialized situations.
\footnote{
One important such specialized situation is SDP, which can be viewed as capturing correlations between pairs of points, in contrast to BLP, which considers only marginal probabilities of individual points.
}

Our first main result, Theorem~\ref{thm:LevelVsPower}, proves that higher levels of \emph{every} minion relaxation can actually be  described by minions. More specifically, for each LC instance $\diag D$, we define (Definition~\ref{def:Spower}) its \emph{$k$-th saturated power $\kthpower{\diag D}$}, another LC instance obtained by introducing a variable for each $k$-tuple of the original LC variables and enforcing all the ``obvious'' constraints. For every minion $\minion{R}$, we then define its \emph{$k$-th level} $\kthlevel{\minion{R}}$. We prove that the $\minion{R}$-relaxation of $\kthpower{\diag D}$ is solvable if and only if the $\kthlevel{\minion{R}}$-relaxation of $\diag D$ is solvable. 

The underlying idea of the $k$-th level construction is to pack enough information about the whole instance to every single relaxed set $\kthlevel{R}$; the challenge is to do this in a well-behaved way. 
More specifically, an object in a relaxed set will contain a fixed part, which describes global information, and a non-fixed part containing more local information. Relaxed maps will ensure that all points of a connected LC instance share the same global information. This seems to be a natural yet powerful design principle: we believe that any reasonable algorithm for CSPs can be characterized by a minion via this idea.

Our construction describes a different version of the $k$-th level algorithm than the traditional CSP-based version in which only $k$-tuples of variables are considered. The reason is that,  by transforming a CSP instance into an LC instance, we place the original variables and the original constraints on equal footing --- our $k$-th level therefore considers $k$-tuples consisting of both variables and constraints.
We view the LC version of higher levels as more natural. One reason is that the basic algorithms admit a simpler and more transparent description from the Label Cover viewpoint --- and to the best of our knowledge, we are not aware of a similar presentation in the literature.

We nevertheless show in Theorem~\ref{thm:VariableLevelsVsConstraintLevels} that the two versions, when applied to CSP instances, are equivalent up to a suitable adjustment of the parameter $k$.
In particular, the applicability of our $k$-th level of 
AC for some $k$ is equivalent to the applicability of the $k$-consistency algorithm for some $k$. In the same sense, our levels 
of BLP capture the levels of the Sherali--Adams 
hierarchy, and our levels of  a combination of AC and AIP capture 
the $\mathbb{Z}$-affine $k$-consistency relaxation proposed by Dalmau 
and Opršal~\cite{dalmauOprsal24reductions} as a candidate uniform 
algorithm (later shown not to be so in~\cite{lichter2024limitations}). 
Finally, the Cohomological $k$-Consistency 
Algorithm~\cite{CohomologicalAlgorithms}, which is another candidate for a uniform algorithm, 
is captured by our levels of the singleton version of a combination of AC and AIP.
Thus, each of the above uniform algorithms is now described by a minion, 
which may facilitate a better understanding of their power.

It is worth mentioning that many papers on hierarchies and their 
limitations have appeared recently, reflecting the growing interest 
in this 
topic~\cite{lichter2024limitations,ChanHawTOFoolHierarchies,ciardo2025semidefinite,
ChanHawTOFoolHierarchiesTwo}. 
The main focus of these papers is to find optimal lower bounds on the 
level required to solve a given problem correctly. 
For instance, Chan, Ng, and Peng~\cite{ChanHawTOFoolHierarchies,
ChanHawTOFoolHierarchiesTwo} showed that if a template satisfies 
certain conditions, then a random CSP instance can fool even a linear 
(in the number of variables) level of several hierarchies. 
In a related direction, Conneryd, Ghannane, and 
Pang~\cite{conneryd2026lower} showed that the cohomological 
$k$-consistency algorithm does not solve the approximate graph coloring 
problem even when $k$ grows linearly with the number of variables.

Among the hierarchies, AIP is perhaps the least well understood. A witness of this is that different papers propose slightly different definitions: some~\cite{berkholz2017linear,ChanHawTOFoolHierarchies} run AIP on partial solutions of the instance, while others~\cite{ciardoZivny2023hierarchies,ciardo2023approximate} introduce integer variables for all evaluations of
$k$-subsets of variables, which are not necessarily partial solutions. It was noted by Jakub Opr\v sal (personal communication) that the latter definition actually makes the hierarchy collapse to the first level, which is a consequence of the main result of~\cite{ciardo2023approximate}. Our paper clarifies this picture: we provide a concrete minion characterizing the $k$-th level, which in particular pins down the right definition, and we exhibit a concrete example (the dihedral group $\mathbf D_{4}$) demonstrating that higher levels of AIP are strictly more powerful than the first.

\subsection{Promise CSPs and relaxations}
\label{sec:intoPCSP}

The significance of minions for CSPs became evident during work on a generalization of fixed-finite-template CSPs called Promise CSPs (PCSPs). 

In this version of the CSP, each constraint relation comes in two forms --  strong and weak. The goal is to distinguish instances whose strong version is solvable from those whose weak version is not, under the promise that one of these two cases holds. 
A \emph{template} of a PCSP can be formally specified as a pair of relational structures $(\str A,\str B)$ of the same signature such that there exists a homomorphism  $\str A \to \str B$ (ensuring that the relations in $\str B$ are indeed weaker). For example, the signature may consist of a ternary symbol $R$, interpreted in $\str A$ as $R^{\str A} = \{001, 010, 100\}$ (the 1-in-3 relation) and in $\str B$ as $R^{\str B} = \{001, 010, 100, 110, 101, 011\}$ (the not-all-equal relation). An instance of the PCSP over $(\str A, \str B)$ is a list of formal constraints, e.g., $xyz \in R$, $yzu \in R$, \dots. An algorithm should accept if  $xyz \in R^{\str A}$, $yzu \in R^{\str A}$, \dots is solvable and reject if not even $xyz \in R^{\str B}$, $yzu \in R^{\str B}$, \dots is solvable. The PCSP over $(\str A,\str B)$ is thus the (positive) 1-in-3-SAT versus not-all-equal-3-SAT problem. 

Inspired by the earlier algebraic theory of fixed-template finite-domain CSPs \cite{AlgebraForCSPBJK,2018wonderland}  and PCSPs \cite{austrin20172+varepsilon,brakensiek2021promise}, 
the authors in \cite{bbko2021} associated to a template $(\str A,\str B)$ a minion  $\minion{M}$,  the \emph{polymorphism minion} of $(\str A,\str B)$, and they proved that the complexity of the PCSP over $(\str A,\str B)$ is (up to log-space reductions) determined by $\minion{M}$. This fact is a consequence of (what we regard as) the \emph{fundamental theorem on fixed-template PCSPs}, which directly links them to relaxations as follows. 

For a minion $\minion{R}$, we denote by $\RelaxLC{\minion{R}}$ the following promise problem: given an LC instance $\diag{D}$, accept if $\diag{D}$ is solvable and reject if the $\minion{R}$-relaxation of $\diag D$ is not solvable. By $\RelaxLCn{\minion{R}}$ we denote the same problem restricted to LC instances whose every domain has size at most $S$. The fundamental theorem says that for every PCSP template $(\str A, \str B)$ and every sufficiently large $S$, the PCSP over $(\str A,\str B)$ is equivalent (wrt. log-space reductions) to $\RelaxLCn{\minion{M}}$, where $\minion{M}$ is the polymorphism minion of $(\str A,\str B)$. 

An immediate consequence of the fundamental theorem is that a relaxation described by a minion $\minion{R}$ correctly decides the PCSP over $(\str A,\str B)$ whenever $\minion{R}$ admits a minion homomorphism to $\minion{M}$. It follows from the known theory that the converse implication holds as well, although we are not aware of an explicit statement of this fact; we provide one  in Theorem~\ref{thm:RelaxationsAndHomos}. By combining this fact with Theorem~\ref{thm:LevelVsPower}, we  obtain that the $k$-th level of the relaxation described by $\minion{R}$ correctly decides the PCSP over $(\str A,\str B)$ if and only if $\kthlevel{\minion{R}}$ admits a homomorphism to $\minion{M}$, see Corollary~\ref{cor:kthLevelChar}. In particular, the applicability of the $k$-th level to solve a given PCSP is completely determined by $\minion{M}$. 

Another immediate consequence of the fundamental theorem is a reduction theorem: if $(\str{A},\str{B})$ and $(\str{A}',\str{B}')$ are PCSP templates with polymorphism minions $\minion{M}$ and $\minion{M}'$, and $\minion{M}$ has a homomorphism to $\minion{M}'$, then the PCSP over $(\str{A}',\str{B})'$ is reducible (in log-space) to the PCSP over $(\str A,\str B)$. While this reduction theorem was, in a sense, sufficient in the context of fixed-template CSPs, it is no longer adequate for PCSPs, and the search for stronger reduction theorems is ongoing \cite{babpcp.datalog,banakh2024injective,LayeredPCP}.  
In \cite{dalmauOprsal24reductions} the authors introduced and studied the so called $k$-consistency reductions. For the smallest level $k=1$, they also provided a characterization in terms of minions: the PCSP over $(\str{A}',\str{B}')$ reduces  to $(\str A,\str B)$ via the  $1$-consistency reduction  if and only if $\omega\minion{M}$, a specific minion constructed from $\minion M$, has a homomorphism to $\minion{M}'$. Using Theorem~\ref{thm:LevelVsPower}, we can now generalize the characterization to all $k$: the condition on minions is that some $\kthlevel{\omega\minion M}$ has a homomorphism to $\minion{M}'$ (Theorem~\ref{thm:CharDatalogReductions}).

\subsection{Solving the square} \label{sec:SolvingSquare}

Now we describe our advertised algorithm to solve the CSP associated with $\mathbf{D}_4$. This problem can be phrased as the CSP over $\str{A}$ (which is the PCSP over $(\str A,\str A)$) for an 8-element structure $\str A$.

Two crucial ideas are involved. The first one is that instead of solving the CSP over $\str A$, we can try to directly solve $\RelaxLCn{\minion{M}}$, where $\minion{M}$ is the polymorphism minion of $\str A$. In fact, our algorithm will solve the search version of $\RelaxLC{\minion{M}}$ in polynomial-time by finding, given an LC instance $\diag{D}$, an explicit solution to the $\minion{M}$-relaxation of $\diag{D}$. This appears to be the first time that the fundamental theorem on PCSPs is used in an algorithmic way.

The minion $\minion{M}$ admits a simple description that follows, e.g., from Lemma 4.1 in \cite{zhuk2025singAlgorithms}. For a finite set $D$, the set $\minion{M}^{(D)}$ consists of all pairs of functions $f\colon D \to \mathbb{Z}_2$, $g\colon D \times D \to \mathbb{Z}_2$ such that 
\begin{equation}\label{eq:definitionOfM}
\sum_{a \in D} f(a) = 1, \quad
\sum_{a,b \in D} g(a,b) = 0, \quad 
(\forall a,b \in D) \ a \neq b \implies g(a,b)+g(b,a) = f(a)f(b),
\end{equation} 
where all computations are performed in the two-element field $\mathbb{Z}_2$. For a function $\alpha\colon D \to E$, the function $\minion{M}^{(\alpha)}$ is defined by summing up. 

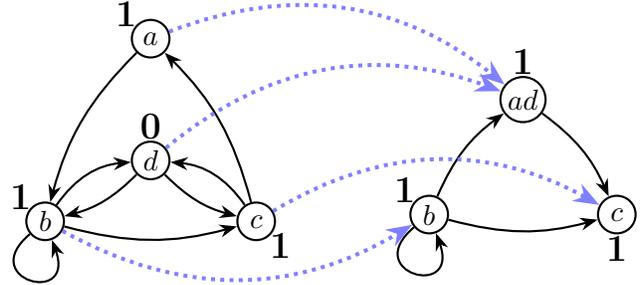
\begin{wrapfigure}{r}{9cm}
\centering
\begin{tikzpicture}[
  every node/.style={circle, draw, thick, minimum size=0.5cm, inner sep=0pt},
  >=Stealth,
  every label/.style={draw=none, circle=none},
trim right = 6.3cm]
\node (A) at (90:1.62)  [label={[label distance=-2pt]above left:\Large\textbf{1}}] {$a$};
\node (B) at (210:1.62) [label={[label distance=-2pt]135:\Large\textbf{1}}]   {$b$};
\node (C) at (330:1.62) [label={[label distance=-2pt]below right:\Large\textbf{1}}] {$c$};
\node (D) at (0,0)      [label={[label distance=-2pt]above:\Large\textbf{0}}] {$d$};
\draw[->, thick, bend left=-15] (A) to (B);
\draw[->, thick, bend left=-15] (B) to (C);
\draw[->, thick, bend left=-15] (C) to (A);
\draw[->, thick] (B) edge [out=220, in=290, looseness=8] (B);
\draw[->, thick, bend right=-15] (D) to (B);
\draw[->, thick, bend left=25]   (B) to (D);
\draw[->, thick, bend left=-15]  (D) to (C);
\draw[->, thick, bend right=25]  (C) to (D);
\begin{scope}[xshift=4.95cm]
  \node (AD) at (0, 0.81) [minimum size=0.6cm, inner sep=0pt,
    label={[label distance=-2pt]above:\Large\textbf{1}}] {$ad$};
  \node (B2) at (210:1.44) [label={[label distance=-2pt]135:\Large\textbf{1}}]   {$b$};
  \node (C2) at (330:1.44) [label={[label distance=-2pt]below:\Large\textbf{1}}] {$c$};
  \draw[->, thick, bend left=15]  (B2) to (AD);
  \draw[<-, thick, bend right=15] (C2) to (AD);
  \draw[->, thick, bend left=-15] (B2) to (C2);
  \draw[->, thick] (B2) edge [out=220, in=290, looseness=8] (B2);
\end{scope}
\draw[->, line width=1.5pt, dotted, blue!50!white, bend left=30]  (A) to (AD);
\draw[->, line width=1.5pt, dotted, blue!50!white, bend left=30]  (D) to (AD);
\draw[->, line width=1.5pt, dotted, blue!50!white, bend left=30]  (C) to (C2);
\draw[->, line width=1.5pt, dotted, blue!50!white, bend right=30] (B) to (B2);
\end{tikzpicture}
\caption{An object from $\minion{M}^{(\{a,b,c,d\})}$ and \\ its  $\minion{M}^{(\alpha)}$-image.}\label{fig:D4GraphExample}
\end{wrapfigure}

An object 
of $\minion{M}^{(D)}$ can be visualized as
a directed graph with vertex set $D$, each marked with $0$ or $1$ (according to $f$), 
whose arcs (according to $g$) satisfy the following properties: 
two $1$-vertices must be connected by exactly one arc; 
other pairs of vertices must be connected by either both arcs or none; 
the total number of $1$-vertices is odd; 
the total number of arcs is even
(see example in Figure~\ref{fig:D4GraphExample}).

The second crucial idea is to combine the SDP and $\mathbb{Z}_2$ relaxations. We find functions $\mathbf{v}_x\colon D_x \to \mathbb{Z}_2^N$ such that $\mathbf{i} = \sum_{a \in D_x} \mathbf{v}_x(a)$ is a unit-weight vector independent of $x$ (where \emph{weight} is the sum of the coordinates modulo 2); for each $x$, the vectors $\mathbf{v}_x(a)$ are pairwise ``orthogonal'' with respect to the standard dot product $\cdot$; and the functions $\mathbf{v}_x$ satisfy the same summation condition as in SDP (except computations take place in the vector space $\mathbb{Z}_2^N$). This relaxation can be computed in polynomial time, as we shall explain shortly.

 Having these vectors, we obtain a solution to the $\minion{M}$-relaxation of $\diag D$ as follows. We choose a bilinear form $\formq\colon \mathbb{Z}_2^N \times \mathbb{Z}_2^N \to \mathbb{Z}_2$ such that for all $\mathbf{v},\mathbf{w} \in \mathbb{Z}_2^N$,
 \begin{equation}\label{eq:formulaForQ}
 \formq(\mathbf{v},\mathbf{w}) + \formq(\mathbf{w},\mathbf{v}) = \weight(\mathbf{v}) \weight(\mathbf{w}) + \mathbf{v} \cdot \mathbf{w},
     \end{equation}
where $\weight$ denotes the weight;
for example, the bilinear form $\formq(v_1v_2\dots v_N,w_1w_2 \dots w_N) = \sum_{i < j} v_iw_j$. 

For every variable $x$, we define $(f_x,g_x) \in \minion{M}^{(\diag D_x)}$ by
$$
f_x(a) = \weight(\mathbf{v}_x(a)), \quad 
g_x(a,a) = \formq(\mathbf{v}_x(a),\mathbf{v}_x(a) - \mathbf{i}), \quad
g_x(a,b) = \formq(\mathbf{v}_x(a),\mathbf{v}_x(b)) \text{ for $a \neq b$}. 
$$

It is straightforward to verify the LC constraints using the bilinearity of $\formq$. For every $x$, the first condition in Equation~(\ref{eq:definitionOfM}) follows from the linearity of $\weight$ and $\weight(\mathbf{i})=1$; the third one from 
Equation (\ref{eq:formulaForQ}) and 
the orthogonality 
$\mathbf{v}_x(a)\cdot \mathbf v_x(b)=0$.
Finally, the bilinearity of $\formq$, combined with the identity $\sum_{a} \mathbf{v}_x(a) = \mathbf{i}$, yields a stronger version of the second condition: for each $a$, $\sum_b g_x(a,b) = 0$.

\subsection{Vector relaxation}

The reason why the above vector relaxation can be computed in polynomial time is that it is equivalent to the second level of the $\mathbb{Z}_2$ relaxation, which can be efficiently solved  by Gaussian elimination.

We now briefly explain the equivalence. Consider an LC instance, such as the one in Figure \ref{fig:2SatDetailedDiagram}. In the vector solution, we attach to each point $p$ (which is, more precisely, a pair $p = xa$ consisting of a variable and an element of its domain)  a vector $\mathbf{v}_p := \mathbf{v}_x(a)$ over $\mathbb{Z}_2$. In the second level $\mathbb{Z}_2$ solution, we attach to each pair of points a number in $\mathbb{Z}_2$; this can be represented by a square matrix $G$ indexed by points. Going from the vector solution to the matrix one is natural: we set $G_{pq} = \mathbf{v}_p \cdot \mathbf{v}_q$; that is, $G$ is the
$\mathbb{Z}_2$ analogue of the Gram matrix of vectors $\mathbf{v}_p$. It turns out that $G$ indeed represents a correct solution. 
Going in the opposite direction, it is not hard to show that, over $\mathbb{Z}_2$ (in contrast to $\mathbb{R}$), every symmetric matrix is a Gram matrix of some vectors (in possibly higher dimension). Such vectors typically will not satisfy the summing constraints, but with some extra care, this idea can be made to work.

For other primes $p$ or larger $k > 2$, every sufficiently symmetric $k$-dimensional tensor is still a Gram tensor (Proposition~\ref{prop:Gram}), but additional complications arise. These can be resolved by the same fixed-part idea as in the $\kthlevel{\minion{R}}$ construction, and we obtain  (Definition~\ref{def:VectorMinion}) a minion $\VkZp$ that characterizes the $k$-th level of the $\mathbb{Z}_p$ relaxation in a satisfactory, SDP-like fashion;  this is our second main result, Theorem~\ref{thm:VkZp}.

We then apply the vector representation to show that the $p$-th level of $\mathbb{Z}_p$ relaxation solves systems of linear equations over $\mathbb Z_{p^2}$, by showing that $\VkZpParam{p}{p}$ has a minion homomorphism to $\mathcal{Z}_{p^2}$ (Proposition~\ref{prop:VkZpToZpSquared}); this appears to be a new result already for $p=2$. The homomorphism is surprisingly simple -- it is enough to sum the entries of the vectors modulo $p^2$. 
This is promising, since a sufficiently general treatment of $\mathbb{Z}_{p^2}$ was one of the key challenges for Zhuk's fixed-template CSP algorithm~\cite{zhukDichotomyJACM}.

We complement this positive result by showing that $\VkZpParam{2}{2}$ does not have a minion homomorphism to $\mathcal{Z}_{8}$ (Proposition~\ref{prop:NotVtwoZtwoToZthree})
and that $\VkZpParam{2}{p}$ does not have a minion homomorphism to $\mathcal{Z}_{p^2}$ for $p>2$ (Proposition~\ref{prop:NoZtwopToZpSquared}).
We do not know whether any level of the $\mathbb{Z}_2$ relaxation 
solves linear equations over $\mathbb{Z}_8$, and we do not know 
the minimal level required to solve linear equations over $\mathbb{Z}_{p^2}$ 
for $p > 3$; we conjecture that it is $p$. This provides us with toy examples for further algorithmic improvements. 
Finally, we note that 
a result of \cite{lichter2023separating} implies that no  fixed level of the $\mathbb{Z}_2$ relaxation solves all $\mathbb{Z}_{2^n}$.

\subsection{Concepts, results, surprises, ideas, and directions}

Our first main conceptual contribution is the Label‑Cover‑based definition of a power of an instance (Definition \ref{def:Spower}), which is easier to analyze than the traditional CSP‑oriented formulations and aligns naturally with standard relaxations.
Our second is the vector minion (Definition \ref{def:VectorMinion}), which
unifies the geometric (SDP) and algebraic ($\mathbb{Z}_p$) aspects within a single object.

Our first main result is the characterization of the $k$-th level of \emph{any} minion-based relaxation via a suitable minion (Theorem~\ref{thm:LevelVsPower}), together with its consequence to fixed-template PCSPs. 
The second main result is the equivalence between higher levels of $\mathbb{Z}_p$ relaxations and vector relaxations (Theorem~\ref{thm:VkZp}), which makes it possible to represent the output of any level of the $\mathbb{Z}_p$ relaxation using objects attached to  \emph{individual} Label Cover points --- a perspective already useful in small cases ($\mathbb{Z}_{p^2}$). 

These results were also surprising to us. It was not expected that higher‑level relaxations could be characterized by a minion---let alone in such generality and in a form that is useful. Equally unexpected was that an AIP‑based algorithm solves $\mathbf{D}_4$: affine relaxations were not meant to handle non‑abelian group CSPs! 

Among the  ideas we highlight are the ``fixed part'' construction, which plays a central role in both main results, and again the SDP–$\mathbb{Z}_p$ combination within a single object (indeed, the complexity of this interaction is one of the reasons the fixed‑template CSP dichotomy was challenging). Another novel idea is the algorithmic use of the fundamental theorem of PCSPs.

Regarding future directions, natural test cases---interesting in their own right---include the commutative rings $\mathbb{Z}_m$ and finite groups such as the dihedral group $\mathbf{D}_8$. Another intriguing challenge is to generalize the second‑level vector algorithm from $\mathbb Z_p$ to $\mathbb{Z}$  (which could also be of interest in the more general Valued PCSP setting \cite{AlgebraicApproachToApprox}). 
At present, we do not know whether the RLC over the corresponding minion is tractable, but note that it is strictly stronger than both SDP and AIP. (We do know, however, that the algorithm itself is not sufficiently powerful, due to an example by Zarathustra Brady.)

A broader goal is to strengthen the approach of \cite{zhuk2025singAlgorithms} to solving fixed-template CSPs, from small domains (7) to general finite domains. With the vector description in hand, this challenge appears more attainable.

Finally, the ideas presented here led us to realize that AIP itself can be enhanced by incorporating a natural form of recursion. This new way of increasing algorithmic power seems genuinely different from both the singleton and higher‑level approaches. Early results look promising: for example, a purely algebraic relaxation obtained in this way is already stronger than singleton AC, which solves all bounded‑width CSPs. We will share these ideas in a separate paper.

\section{Preliminaries}

Our formalism and results are largely inspired by the categorical perspective on the CSP presented in \cite{hadekJaklOprsal2026categories}. In this paper, we have chosen a middle ground between the categorical approach and the more traditional ones.

We introduce only those concepts that are needed in Sections~\ref{sec:levels}, \ref{sec:vectors}, and \ref{sec:ProofOfLevels}. The remaining concepts, in particular the PCSPs and their polymorphism minions, are introduced in Section~\ref{sec:PCSP}.

\begin{defin}[LC instance] \label{def:LCinstance}
 An \emph{LC instance} $\diag D$ consists of a finite set $X$ of variables, 
 finite set $\diag D_x$ for each variable $x \in X$ called the \emph{domain} of $x$, and a finite list of constraints of the form $x = \alpha(y)$, where $x,y \in X$ and $\alpha\colon \diag D_y \to \diag D_x$. 

A \emph{point} in $\diag D$ is a pair $xa$, with $x \in X$ and $a \in D_x$.

A \emph{solution} to $\diag D$ is a tuple $(d_x)_{x \in X}$, with $d_x\in\diag D_x$ for all variables $x$, satisfying $d_x = \alpha(d_y)$ for every constraint $x = \alpha(y)$ in $\diag D$.
 \end{defin}

\begin{defin}
    The \emph{shape} of an LC instance $\diag D$ is the directed multigraph whose vertices are the variables, with an arc  $x \leftarrow y$ for each constraint 
    $x = \alpha(y)$ in $\diag D$. 

    $\diag D$ is \emph{connected} if its shape is connected.

    $\diag D$ is \emph{discrete} if $\diag D$ has no constraints.
\end{defin}

Until Section~\ref{sec:PCSP}, (abstract) minions play the role of describing relaxations. 
In the language of category theory, a  minion is a functor from the category of finite sets to the category of sets. We spell this out in detail as follows.

\begin{definition}[Minion]
A \emph{minion} $\minion R$ consists of a collection of sets $\minion R^{(D)}$, indexed by all finite sets $D$, together with a \emph{minor map} ${\minion R}^{(\alpha)}\colon  {\minion R}^{(D)}\to  {\minion R}^{(E)}$ for every function $\alpha\colon D\to E$, which satisfies that ${\minion R}^{(\id_{D})}= \id_{{\minion R}^{(D)}}$ for all finite sets $D$ and ${\minion R}^{(\alpha)} \circ {\minion R}^{(\beta)} = \minion R^{(\alpha \circ \beta)}$ whenever such a composition is well-defined. 
\end{definition} 

\begin{defin}
    Let $\diag D$ be an LC instance and $\minion{R}$ a minion. The \emph{$\minion{R}$-relaxation of $\diag{D}$}, denoted $\minion{R} \circ \diag D$, is obtained from $\minion{D}$ by replacing each domain $\minion{D}_x$ with $\minion{R}^{(\diag D_x)}$ and each constraint $x = \alpha(y)$ with $x = \minion R^{(\alpha)}(y)$.

    A solution of $\minion{R} \circ \diag D$ is also referred to as a \emph{solution to $\diag D$ in $\minion R$}, or an \emph{$\minion{R}$-solution to $\diag D$}.
\end{defin}

The most important minions for us are $\mathcal{Z}_n$, in particular when $n$ is a prime number.

\begin{defin}
   Let $n \in \mathbb{N}$. The minion $\mathcal{Z}_n$ is defined by $\mathcal{Z}_n^{(D)} = \{f\colon D \to \mathbb{Z}_n \mid \sum_{a \in D} f(a) = 1\}$ for every finite set $D$, and $\mathcal{Z}_n^{(\alpha)}(f) = \summing{\alpha}{f}$, where $(\summing{\alpha}{f})(b) = \sum_{a \in \alpha^{-1}(b)} f(a)$, for every function $\alpha\colon D \to E$.  
\end{defin}

Minions can be compared by means of minion homomorphisms.

\begin{defin} Let $\minion R $ and $\minion S$ be minions.
    A \emph{minion homomorphism} from $\minion{R} $ to  $\minion{S}$ is a collection of functions $\xi^{(D)}\colon \minion R^{(D)} \to \minion S^{(D)}$, indexed by all finite sets $D$, that commute with minor maps; that is, $\xi^{(E)} \circ \minion R^{(\alpha)}={\minion S}^{(\alpha)} \circ \xi^{(D)}$ for every function $\alpha\colon D \to E$.
\end{defin}

\begin{prop} \label{prop:HomosPreserveSolutions}
    Let $\diag D$ be an LC instance, and $\minion{R}$ and $\minion{S}$ be minions with a minion homomorphism $\xi\colon \minion{R} \to \minion{S}$. If $\diag D$ has a solution in $\minion{R}$, then it has a solution in $\minion{S}$.
\end{prop}
\begin{proof}[Proof sketch]
If $(d_x)_{x \in X}$ is a solution in $\minion{R}$, then $(\xi^{\diag D_x}(d_x))_{x \in X}$ is a solution in $\minion{S}$.
\end{proof}

Finally, post-composition is denoted $\cdot^*$, that is $q^*(d) = d \circ q$. For a set $A$, $A$-tuples are formally functions from $A$; $n$-tuples (that is, $[n]$-tuples) are written as $a_1a_2\dots a_n$ or $(a_1,a_2,\dots, a_n)$. Conversely, functions from $[n]$ are sometimes written as tuples, especially when specifying $\alpha$ in a minor map, or when working with post-composition. For example,
$$
(4321)^* (abcd) = dcba, \quad (1132)^*(abc) = (aacb).
$$

\section{Higher levels} \label{sec:levels}

\subsection{Saturated power of an LC instance}

The saturated power of an LC instance $\diag D$ is an LC instance whose variables are all $k$-tuples of variables from the original instance, and whose constraints are those ``obviously'' implied by the original constraints. There are two types of such constraints. The first type arises from product maps: for example, if $\diag{D}$ contains $x=\alpha(y)$ and $x' = \alpha'(y')$, then the second saturated power contains $xx' = \alpha \times \alpha'(yy')$, where $\alpha \times \alpha' (dd') = \alpha(d)\alpha(d')$. Since we also want to include constraints such as $xz = \alpha \times \id (yz)$,  it is convenient to add the trivial constraints such as $z = \id(z)$ to $\diag D$.
The second type consists of constraints obtained by permuting or merging coordinates, e.g.,  $xy = (2 1)^*(yx)$ (recall that $(2 1)^*$ maps $ab$ to $ba$) and $xx = (1 1)^*(xy)$ (where $(11)^*$ maps $ab$ to $aa$).

\begin{defin}[Saturated power] \label{def:Spower}
    Let $\diag D$ be an LC instance as in Definition~\ref{def:LCinstance} and let $k \in \mathbb N$. The \emph{$k$-th saturated power of $\diag D$}, denoted $\kthpower{\diag D}$, has variable set $X^k$ and domains 
$$    
\diag D_{x_1x_2\dots x_k} = \diag D_{x_1} \times \diag D_{x_2} \times \cdots \times \diag D_{x_k}. 
$$ To define the constraints, we first add to $\diag D$ all constraints of the form $x = \id(x)$, and then include in $\kthpower{\diag D}$ the following two types of constraints. 
    \begin{enumerate}
        \item[(1)]  For every $k$-tuple of $\diag D$-constraints $x_i = \alpha_i(y_i)$, $i \in [k]$, we include the constraint $x_1x_2\dots x_k = \alpha_1\times \alpha_2 \times \dots \times\alpha_k \;(y_1y_2 \dots y_k)$:
    \begin{align*}
       \diag{D}_{y_1y_2\dots y_k}&\xrightarrow{\alpha_1\times \alpha_2 \times \dots \times\alpha_k}\diag{D}_{x_1x_2\dots x_k} \\       
        d_1d_2\dots d_k &\xmapsto{\phantom{\alpha_1\times \alpha_2 \times \dots \times\alpha_k}}
        \alpha_1(d_1)\alpha_2(d_2) \dots \alpha_k(d_k)
    \end{align*}
        \item[(2)] For every $x_1x_2 \dots x_k \in X^k$ and every $\sigma\colon [k] \to [k]$, we include the constraint $x_{\sigma(1)}x_{\sigma(2)} \dots x_{\sigma(k)} = \sigma^{*}(x_1x_2 \dots x_k)$:
    \begin{align*}
    \diag{D}_{x_1x_2\dots x_k} & \xrightarrow{\sigma^*} \diag{D}_{x_{\sigma(1)}x_{\sigma(2)}\dots x_{\sigma(k)}}\\
    d_1d_2\dots d_k &\xmapsto{\phantom{\sigma^*}}
    d_{\sigma(1)}d_{\sigma(2)} \dots d_{\sigma(k)} =d\circ \sigma 
    \end{align*}
    \end{enumerate}
\end{defin}

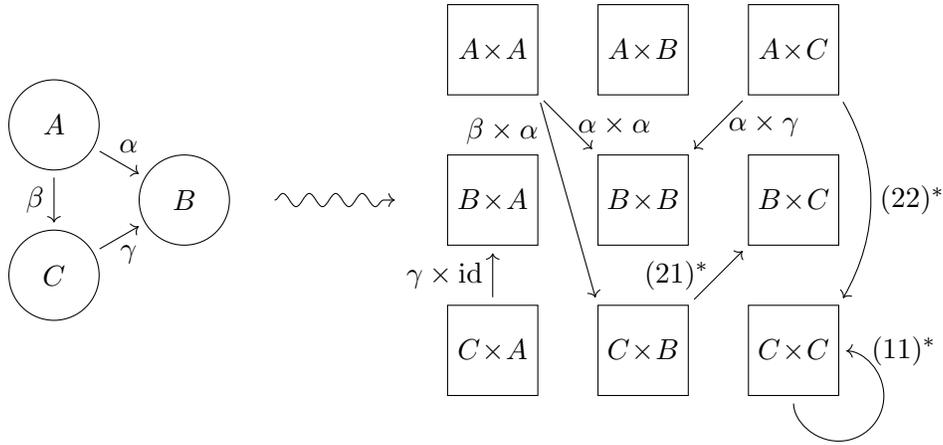
\begin{figure}[h]
    \centering
\begin{tikzpicture}[
  box/.style={draw, minimum width=1.2cm, minimum height=1.2cm,
              inner sep=0pt, font=\normalsize},
  barr/.style={-{Latex}, line width=1.5pt, blue, shorten >=0.1cm, shorten <=0.1cm}
]
    \def\r{0.6}
    \def\d{2}
    \def\t{150}
    \def\s{1.2}
    \def\center{(-1.5,-2*\d-\s/2)}
    \draw[->,decoration={snake},decorate] (-0.3,-2*\d-\s/2) -- (1.3,-2*\d-\s/2);
    \draw [shift=\center](\t:\d) coordinate (o1) circle (\r);
        \node at (o1) {$A$};
    \draw [shift=\center](\t+60:\d) coordinate (o2) circle (\r);
        \node at (o2) {$C$};
    \draw \center coordinate (o3) circle (\r);
        \node at (o3) {$B$};
    \draw[->, shorten > = 0.7cm, shorten < = 0.7cm] (o1) -- (o2) node[midway, left] {$\beta$};
    \draw[->, shorten > = 0.7cm, shorten < = 0.7cm] (o1) -- (o3) coordinate[midway] (c);
    \node at ([shift={(60:0.25)}]c) {$\alpha$};
    \draw[->, shorten > = 0.7cm, shorten < = 0.7cm] (o2) -- (o3) coordinate[midway] (c);
    \node at ([shift={(120:-0.25)}]c) {$\gamma$};
    \node[box] (AA) at (2.6,-2.6) {$A\!\times\! A$};
    \node[box] (AB) at (4.6,-2.6) {$A\!\times\! B$};
    \node[box] (AC) at (6.6,-2.6) {$A\!\times\! C$};
    \node[box] (BA) at (2.6,-4.6) {$B\!\times\! A$};
    \node[box] (BB) at (4.6,-4.6) {$B\!\times\! B$};
    \node[box] (BC) at (6.6,-4.6) {$B\!\times\! C$};
    \node[box] (CA) at (2.6,-6.6) {$C\!\times\! A$};
    \node[box] (CB) at (4.6,-6.6) {$C\!\times\! B$};
    \node[box] (CC) at (6.6,-6.6) {$C\!\times\! C$};
    \draw[->, shorten >=0.1cm, shorten <=0.1cm]
        (AA.south east) -- (BB.north west)
        node[midway, right] {$\alpha\times \alpha$};
    \draw[->, shorten >=0.1cm, shorten <=0.1cm]
        (AA.south east) -- (CB.north west)
        node[pos=0.17, left] {$\beta\times \alpha$};
    \draw[->, shorten >=0.1cm, shorten <=0.1cm]
        (CA.north) -- (BA.south)
        node[midway, left] {$\gamma\times \mathrm{id}$};
    \draw[->, shorten >=0.1cm, shorten <=0.1cm]
        (AC.south west) -- (BB.north east)
        node[midway, right] {$\alpha\times\gamma$};
    \draw[->, shorten >=0.1cm, shorten <=0.1cm]
        (CB.north east) -- (BC.south west)
        node[midway, left] {$(21)^*$};
    \draw[->, shorten >=0.1cm, shorten <=0.1cm]
        (AC.south east) to[bend left] node[midway, right] {$(22)^*$} (CC.north east);
    \draw[->, shorten >=0.1cm, shorten <=0.1cm]
        (CC.south) arc (180:450:0.6cm)
        node[xshift = 0.3cm, right] {$(11)^*$};
\end{tikzpicture}
    \caption{Saturated power}
    \label{fig:SPower}
\end{figure}

Note that the saturated power is nontrivial even for discrete $\diag D$; it contains the constraints of type (2) enforcing consistency among permutations of tuples of variables and variable merges.

Observe also that there are redundancies. For example, it would be enough to include type (2) constraints and type (1) constraints whose product maps have the form $\alpha \times \id \times \id \times \cdots \times \id$.  

\subsection{Second level example}

We now discuss the $\mathbb{Z}_2$ relaxation and its second level in some detail. 
Recall that the $\mathbb{Z}_2$ relaxation is described by the minion $\Ztwo$ with  $\Ztwo^{(D)} = \{f\colon D \to \mathbb{Z}_2\mid\sum_{a \in D} f(a)=1\}$ and the functions $\Ztwo^{(\alpha)}$ defined in the natural way -- by summation. Looking at Figure~\ref{fig:SPower}, a solution of $\Ztwo \circ \diag D$, the $\Ztwo$-relaxation  of $\diag D$, assigns a number in $\mathbb{Z}_2$ to every point on the left, in a way consistent with the constraints and definitions of $\Ztwo^{(D)}$ and $\Ztwo^{(\alpha)}$. 
Likewise, a solution of $\Ztwo \circ \secondpower{\diag D}$, the second level $\Ztwo$-relaxation of $\diag D$, assigns a number to every point on the right (although the points are not shown in the figure). 

To give a specific example, we consider the LC instance $\diag D$ associated with the CSP instance
$$
 \varx\le \vary, \  
 \vary\le \varz, \ 
 \varz\le \varu, \  
 \varu<\varx \quad
 \text{ over the domain $\{0,1\}$}.
$$

A solution of $\Ztwo \circ \diag{D}$ assigns to every point $xa$ consisting of an LC variable $x$, which is an original CSP variable or an original constraint, and a domain element $a \in \diag D_x$, a value $f_x(a)$ in $\mathbb{Z}_2$.
The LC instance $\diag D$ together with its $\Ztwo$-solution is shown in Figure \ref{fig:ZtwoInstance}.

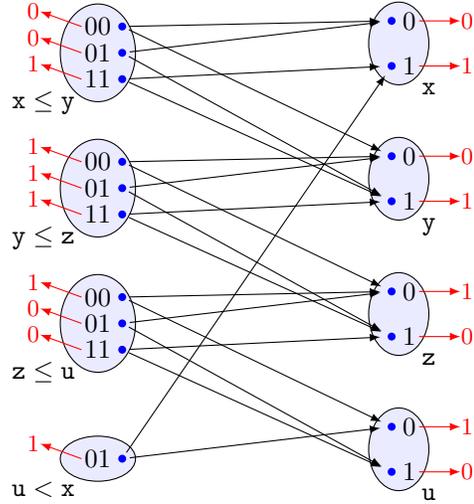
\begin{wrapfigure}{r}{0.45\textwidth}
    \centering
\begin{tikzpicture}[
    tuple/.style={inner sep=1pt, font=\small},
    label/.style={font=\small, inner sep=0pt, anchor=north east},
    rlabel/.style={font=\small, inner sep=0pt, anchor=north west},
    aip/.style={color=red, font=\footnotesize, inner sep=0pt}
]
\begin{scope}
  \clip (0, 2.7) ellipse (0.50cm and 0.65cm);
  \fill[blue!8!white] (-1, 2.05) rectangle (1, 3.35);
\end{scope}
\draw (0, 2.7) ellipse (0.50cm and 0.65cm);
\node[tuple] (Rxy00) at (0, 3.05) {$00$};
\node[tuple] (Rxy01) at (0, 2.70) {$01$};
\node[tuple] (Rxy11) at (0, 2.35) {$11$};
\node[label] at (-0.3, 2.20) {$\varx\le \vary$};

\begin{scope}
  \clip (0, 0.9) ellipse (0.50cm and 0.65cm);
  \fill[blue!8!white] (-1, 0.25) rectangle (1, 1.55);
\end{scope}
\draw (0, 0.9) ellipse (0.50cm and 0.65cm);
\node[tuple] (Ryz00) at (0, 1.25) {$00$};
\node[tuple] (Ryz01) at (0, 0.90) {$01$};
\node[tuple] (Ryz11) at (0, 0.55) {$11$};
\node[label] at (-0.3, 0.4) {$\vary\le \varz$};

\begin{scope}
  \clip (0, -0.9) ellipse (0.50cm and 0.65cm);
  \fill[blue!8!white] (-1, -1.55) rectangle (1, -0.25);
\end{scope}
\draw (0, -0.9) ellipse (0.50cm and 0.65cm);
\node[tuple] (Rzu00) at (0, -0.55) {$00$};
\node[tuple] (Rzu01) at (0, -0.90) {$01$};
\node[tuple] (Rzu11) at (0, -1.25) {$11$};
\node[label] at (-0.3, -1.40) {$\varz\le \varu$};

\begin{scope}
  \clip (0, -2.7) ellipse (0.50cm and 0.30cm);
  \fill[blue!8!white] (-1, -3.00) rectangle (1, -2.40);
\end{scope}
\draw (0, -2.7) ellipse (0.50cm and 0.30cm);
\node[tuple] (Rux01) at (0, -2.70) {$01$};
\node[label] at (-0.3, -3) {$\varu< \varx$};

\begin{scope}
  \clip (4, 2.825) ellipse (0.40cm and 0.55cm);
  \fill[blue!8!white] (3, 2.275) rectangle (5, 3.375);
\end{scope}
\draw (4, 2.825) ellipse (0.40cm and 0.55cm);
\node[tuple] (x0) at (4, 3.125) [xshift=4pt] {$0$};
\node[tuple] (x1) at (4, 2.525) [xshift=4pt] {$1$};
\node[rlabel] at (4.30, 2.305) {$\varx$};

\begin{scope}
  \clip (4, 1.025) ellipse (0.40cm and 0.55cm);
  \fill[blue!8!white] (3, 0.475) rectangle (5, 1.575);
\end{scope}
\draw (4, 1.025) ellipse (0.40cm and 0.55cm);
\node[tuple] (y0) at (4, 1.325) [xshift=4pt] {$0$};
\node[tuple] (y1) at (4, 0.725) [xshift=4pt] {$1$};
\node[rlabel] at (4.30, 0.505) {$\vary$};

\begin{scope}
  \clip (4, -0.775) ellipse (0.40cm and 0.55cm);
  \fill[blue!8!white] (3, -1.325) rectangle (5, -0.225);
\end{scope}
\draw (4, -0.775) ellipse (0.40cm and 0.55cm);
\node[tuple] (z0) at (4, -0.475) [xshift=4pt] {$0$};
\node[tuple] (z1) at (4, -1.075) [xshift=4pt] {$1$};
\node[rlabel] at (4.30, -1.295) {$\varz$};

\begin{scope}
  \clip (4, -2.575) ellipse (0.40cm and 0.55cm);
  \fill[blue!8!white] (3, -3.125) rectangle (5, -2.025);
\end{scope}
\draw (4, -2.575) ellipse (0.40cm and 0.55cm);
\node[tuple] (u0) at (4, -2.275) [xshift=4pt] {$0$};
\node[tuple] (u1) at (4, -2.875) [xshift=4pt] {$1$};
\node[rlabel] at (4.30, -3.095) {$\varu$};

\draw[-latex, shorten <=0.1cm, shorten >=0.15cm] ([xshift=3pt]Rxy00.east) to ([xshift=-3pt]x0.west);
\draw[-latex, shorten <=0.1cm, shorten >=0.15cm] ([xshift=3pt]Rxy00.east) to ([xshift=-3pt]y0.west);
\draw[-latex, shorten <=0.1cm, shorten >=0.15cm] ([xshift=3pt]Rxy01.east) to ([xshift=-3pt]x0.west);
\draw[-latex, shorten <=0.1cm, shorten >=0.15cm] ([xshift=3pt]Rxy01.east) to ([xshift=-3pt]y1.west);
\draw[-latex, shorten <=0.1cm, shorten >=0.15cm] ([xshift=3pt]Rxy11.east) to ([xshift=-3pt]x1.west);
\draw[-latex, shorten <=0.1cm, shorten >=0.15cm] ([xshift=3pt]Rxy11.east) to ([xshift=-3pt]y1.west);
\draw[-latex, shorten <=0.1cm, shorten >=0.15cm] ([xshift=3pt]Ryz00.east) to ([xshift=-3pt]y0.west);
\draw[-latex, shorten <=0.1cm, shorten >=0.15cm] ([xshift=3pt]Ryz00.east) to ([xshift=-3pt]z0.west);
\draw[-latex, shorten <=0.1cm, shorten >=0.15cm] ([xshift=3pt]Ryz01.east) to ([xshift=-3pt]y0.west);
\draw[-latex, shorten <=0.1cm, shorten >=0.15cm] ([xshift=3pt]Ryz01.east) to ([xshift=-3pt]z1.west);
\draw[-latex, shorten <=0.1cm, shorten >=0.15cm] ([xshift=3pt]Ryz11.east) to ([xshift=-3pt]y1.west);
\draw[-latex, shorten <=0.1cm, shorten >=0.15cm] ([xshift=3pt]Ryz11.east) to ([xshift=-3pt]z1.west);
\draw[-latex, shorten <=0.1cm, shorten >=0.15cm] ([xshift=3pt]Rzu00.east) to ([xshift=-3pt]z0.west);
\draw[-latex, shorten <=0.1cm, shorten >=0.15cm] ([xshift=3pt]Rzu00.east) to ([xshift=-3pt]u0.west);
\draw[-latex, shorten <=0.1cm, shorten >=0.15cm] ([xshift=3pt]Rzu01.east) to ([xshift=-3pt]z0.west);
\draw[-latex, shorten <=0.1cm, shorten >=0.15cm] ([xshift=3pt]Rzu01.east) to ([xshift=-3pt]u1.west);
\draw[-latex, shorten <=0.1cm, shorten >=0.15cm] ([xshift=3pt]Rzu11.east) to ([xshift=-3pt]z1.west);
\draw[-latex, shorten <=0.1cm, shorten >=0.15cm] ([xshift=3pt]Rzu11.east) to ([xshift=-3pt]u1.west);
\draw[-latex, shorten <=0.1cm, shorten >=0.15cm] ([xshift=3pt]Rux01.east) to ([xshift=-3pt]u0.west);
\draw[-latex, shorten <=0.1cm, shorten >=0.15cm] ([xshift=3pt]Rux01.east) to ([xshift=-3pt]x1.west);

\fill[blue] ([xshift=3pt]Rxy00.east) circle (1.5pt);
\fill[blue] ([xshift=3pt]Rxy01.east) circle (1.5pt);
\fill[blue] ([xshift=3pt]Rxy11.east) circle (1.5pt);
\fill[blue] ([xshift=3pt]Ryz00.east) circle (1.5pt);
\fill[blue] ([xshift=3pt]Ryz01.east) circle (1.5pt);
\fill[blue] ([xshift=3pt]Ryz11.east) circle (1.5pt);
\fill[blue] ([xshift=3pt]Rzu00.east) circle (1.5pt);
\fill[blue] ([xshift=3pt]Rzu01.east) circle (1.5pt);
\fill[blue] ([xshift=3pt]Rzu11.east) circle (1.5pt);
\fill[blue] ([xshift=3pt]Rux01.east) circle (1.5pt);
\fill[blue] ([xshift=-3pt]x0.west) circle (1.5pt);
\fill[blue] ([xshift=-3pt]x1.west) circle (1.5pt);
\fill[blue] ([xshift=-3pt]y0.west) circle (1.5pt);
\fill[blue] ([xshift=-3pt]y1.west) circle (1.5pt);
\fill[blue] ([xshift=-3pt]z0.west) circle (1.5pt);
\fill[blue] ([xshift=-3pt]z1.west) circle (1.5pt);
\fill[blue] ([xshift=-3pt]u0.west) circle (1.5pt);
\fill[blue] ([xshift=-3pt]u1.west) circle (1.5pt);

\draw[-latex, red] (Rxy00.west) -- ++(-0.55, 0.2) node[left, aip] {$0$};
\draw[-latex, red] (Rxy01.west) -- ++(-0.55, 0.2) node[left, aip] {$0$};
\draw[-latex, red] (Rxy11.west) -- ++(-0.55, 0.2) node[left, aip] {$1$};
\draw[-latex, red] (Ryz00.west) -- ++(-0.55, 0.2) node[left, aip] {$1$};
\draw[-latex, red] (Ryz01.west) -- ++(-0.55, 0.2) node[left, aip] {$1$};
\draw[-latex, red] (Ryz11.west) -- ++(-0.55, 0.2) node[left, aip] {$1$};
\draw[-latex, red] (Rzu00.west) -- ++(-0.55, 0.2) node[left, aip] {$1$};
\draw[-latex, red] (Rzu01.west) -- ++(-0.55, 0.2) node[left, aip] {$0$};
\draw[-latex, red] (Rzu11.west) -- ++(-0.55, 0.2) node[left, aip] {$0$};
\draw[-latex, red] (Rux01.west) -- ++(-0.55, 0.2) node[left, aip] {$1$};

\draw[-latex, red] (x0.east) -- ++(0.55, 0) node[right, aip] {$0$};
\draw[-latex, red] (x1.east) -- ++(0.55, 0) node[right, aip] {$1$};
\draw[-latex, red] (y0.east) -- ++(0.55, 0) node[right, aip] {$0$};
\draw[-latex, red] (y1.east) -- ++(0.55, 0) node[right, aip] {$1$};
\draw[-latex, red] (z0.east) -- ++(0.55, 0) node[right, aip] {$1$};
\draw[-latex, red] (z1.east) -- ++(0.55, 0) node[right, aip] {$0$};
\draw[-latex, red] (u0.east) -- ++(0.55, 0) node[right, aip] {$1$};
\draw[-latex, red] (u1.east) -- ++(0.55, 0) node[right, aip] {$0$};
\end{tikzpicture}
    \caption{Solution of the $\mathbb{Z}_2$ relaxation.}
    \label{fig:ZtwoInstance}
\end{wrapfigure}
 
Observe that the condition in the definition of $\Ztwo^{(D)}$ simply means that the values assigned to the elements of each blue set sum to 1. Moreover, the definition of $\Ztwo^{(\alpha)}$ translates to the requirement that the value assigned to any element on the right is the sum of values assigned to its neighbors on the left. 

A solution of $\Ztwo \circ \secondpower{\diag D}$  assigns a value in $\mathbb{Z}_2$ to every pair of points $(xa, yb)$ --  the value $f_{xy}(ab)$. Such a solution is shown in Figure~\ref{fig:SecondLevelSolution}: the value $f_{xy}(ab)$ is in the row $xa$ and column $yb$. This time, the definition of $\Ztwo^{(D)}$ ensures that the values in each block of the matrix sum to 1. The constraints corresponding to $(21)^{*}$  ensure that the matrix is symmetric.  The constraints corresponding to $(11)^{*}$ ensure that each diagonal block is a diagonal matrix and that, within each block, every row sums to the diagonal entry of the same row. Finally, the constraints corresponding to $\alpha \times \id$  impose conditions on entire rows; for example, 
the row $\vary,1$ is the sum of the rows $\varx \leq \vary,01$ and $\varx \leq \vary,11$.

\begin{figure}[h]
{\renewcommand{\arraystretch}{1.45}
$$\begin{array}{r!{\vrule width 2pt}*{2}{w{c}{0.9em}}|*{2}{w{c}{0.9em}}|*{2}{w{c}{0.9em}}|*{2}{w{c}{0.9em}}|*{3}{w{c}{0.9em}}|*{3}{w{c}{0.9em}}|*{3}{w{c}{0.9em}}|w{c}{0.9em}!{\vrule width 2pt}}
 & \rotatebox{90}{$\varx,0$} & \rotatebox{90}{$\varx,1$} & \rotatebox{90}{$\vary,0$} & \rotatebox{90}{$\vary,1$} & \rotatebox{90}{$\varz,0$} & \rotatebox{90}{$\varz,1$} & \rotatebox{90}{$\varu,0$} & \rotatebox{90}{$\varu,1$} & \rotatebox{90}{$\varx\le \vary,00$} & \rotatebox{90}{$\varx\le \vary,01$} & \rotatebox{90}{$\varx\le \vary,11$} & \rotatebox{90}{$\vary\le \varz,00$} & \rotatebox{90}{$\vary\le \varz,01$} & \rotatebox{90}{$\vary\le \varz,11$} & \rotatebox{90}{$\varz\le \varu,00$} & \rotatebox{90}{$\varz\le \varu,01$} & \rotatebox{90}{$\varz\le \varu,11$} & \rotatebox{90}{$\varu<\varx,01$} \\
\noalign{\hrule height 2pt}
\varx,0 & \cellcolor{blue!10}0&\cellcolor{blue!10}0&0&0&0&0&0&0&0&0&0&0&0&0&0&0&0&0\\
\varx,1 & \cellcolor{blue!10}0&\cellcolor{blue!10}1&0&1&1&0&1&0&0&0&1&1&1&1&1&0&0&1\\
\hline
\vary,0 & 0&0&\cellcolor{blue!10}0&\cellcolor{blue!10}0&1&1&0&0&0&0&0&1&1&0&0&1&1&0\\
\vary,1 & 0&1&\cellcolor{blue!10}0&\cellcolor{blue!10}1&0&1&1&0&0&0&1&0&0&1&1&1&1&1\\
\hline
\varz,0 & 0&1&1&0&\cellcolor{blue!10}1&\cellcolor{blue!10}0&1&0&1&1&1&1&0&0&1&0&0&1\\
\varz,1 & 0&0&1&1&\cellcolor{blue!10}0&\cellcolor{blue!10}0&0&0&1&1&0&0&1&1&0&0&0&0\\
\hline
\varu,0 & 0&1&0&1&1&0&\cellcolor{blue!10}1&\cellcolor{blue!10}0&0&0&1&1&1&1&1&0&0&1\\
\varu,1 & 0&0&0&0&0&0&\cellcolor{blue!10}0&\cellcolor{blue!10}0&0&0&0&0&0&0&0&0&0&0\\
\hline
\varx\le \vary,00 & 0&0&0&0&1&1&0&0&\cellcolor{blue!10}0&\cellcolor{blue!10}0&\cellcolor{blue!10}0&1&1&0&0&1&1&0\\
\varx\le \vary,01 & 0&0&0&0&1&1&0&0&\cellcolor{blue!10}0&\cellcolor{blue!10}0&\cellcolor{blue!10}0&1&1&0&0&1&1&0\\
\varx\le \vary,11 & 0&1&0&1&1&0&1&0&\cellcolor{blue!10}0&\cellcolor{blue!10}0&\cellcolor{blue!10}1&1&1&1&1&0&0&1\\
\hline
\vary\le \varz,00 & 0&1&1&0&1&0&1&0&1&1&1&\cellcolor{blue!10}1&\cellcolor{blue!10}0&\cellcolor{blue!10}0&1&0&0&1\\
\vary\le \varz,01 & 0&1&1&0&0&1&1&0&1&1&1&\cellcolor{blue!10}0&\cellcolor{blue!10}1&\cellcolor{blue!10}0&1&1&1&1\\
\vary\le \varz,11 & 0&1&0&1&0&1&1&0&0&0&1&\cellcolor{blue!10}0&\cellcolor{blue!10}0&\cellcolor{blue!10}1&1&1&1&1\\
\hline
\varz\le \varu,00 & 0&1&0&1&1&0&1&0&0&0&1&1&1&1&\cellcolor{blue!10}1&\cellcolor{blue!10}0&\cellcolor{blue!10}0&1\\
\varz\le \varu,01 & 0&0&1&1&0&0&0&0&1&1&0&0&1&1&\cellcolor{blue!10}0&\cellcolor{blue!10}0&\cellcolor{blue!10}0&0\\
\varz\le \varu,11 & 0&0&1&1&0&0&0&0&1&1&0&0&1&1&\cellcolor{blue!10}0&\cellcolor{blue!10}0&\cellcolor{blue!10}0&0\\
\hline
\varu<\varx,01 & 0&1&0&1&1&0&1&0&0&0&1&1&1&1&1&0&0&\cellcolor{blue!10}1\\
\noalign{\hrule height 2pt}
\end{array}$$}
\caption{A solution to the second level of the $\mathbb{Z}_2$ relaxation.}
\label{fig:SecondLevelSolution}
\end{figure}

\subsection{Levels of a minion}

We arrive at the definition of the $k$-th level of a minion $\minion{R}$. An object of the ``relaxed set $A$'' of this minion --- illustrated for $k=2$ in Figure \ref{fig:kthlevel} --- consists of a \emph{fixed part}, which does not change under any minor map, and a non-fixed part. The fixed part is an $\minion{R}$-solution $r$ of the $k$-th saturated power of an arbitrary LC instance $\diag D$. The non-fixed part is a ``row minor'' of a ``row block'' of $r$. 
To gain intuition for $\minion{R} = \minion{Z}_2$, it may also be helpful to revisit Figure \ref{fig:SecondLevelSolution}.

\begin{figure}[h]
\centering
\begin{tikzpicture}[
  box/.style={draw, minimum width=1.2cm, minimum height=1.2cm,
              inner sep=0pt, font=\small},
  abox/.style={draw, minimum width=1.2cm, minimum height=0.84cm,
               inner sep=0pt, font=\small},
  barr/.style={-{Latex}, line width=1.5pt, blue, shorten >=0.1cm, shorten <=0.1cm}
]
\node[box] (D1D1) at (2.6,-2.6) {$D_1\!\times\! D_1$};
\node[box] (D1D2) at (4.6,-2.6) {$D_1\!\times\! D_2$};
\node[box] (D1D3) at (6.6,-2.6) {$D_1\!\times\! D_3$};
\node[box] (D2D1) at (2.6,-4.6) {$D_2\!\times\! D_1$};
\node[box] (D2D2) at (4.6,-4.6) {$D_2\!\times\! D_2$};
\node[box] (D2D3) at (6.6,-4.6) {$D_2\!\times\! D_3$};
\node[box] (D3D1) at (2.6,-6.6) {$D_3\!\times\! D_1$};
\node[box] (D3D2) at (4.6,-6.6) {$D_3\!\times\! D_2$};
\node[box] (D3D3) at (6.6,-6.6) {$D_3\!\times\! D_3$};
\draw[decorate, decoration={snake, amplitude=2pt, segment length=8pt}]
    (1.8,-8.1) -- (7.4,-8.1);
\node[abox] (AD1) at (2.6,-9.0) {$A\!\times\! D_1$};
\node[abox] (AD2) at (4.6,-9.0) {$A\!\times\! D_2$};
\node[abox] (AD3) at (6.6,-9.0) {$A\!\times\! D_3$};
\draw[->, shorten >=0.1cm, shorten <=0.1cm]
    (D2D1.north east) -- (D1D2.south west)
    node[midway, left] {$(21)^*$};
\draw[->, shorten >=0.1cm, shorten <=0.1cm]
    (D1D3.south east) to[bend left] node[midway, right] {$(22)^*$} (D3D3.north east);
\draw[->, shorten >=0.1cm, shorten <=0.1cm]
    (D3D3.south) arc (180:450:0.6cm)
    node[xshift=1cm, yshift=-0.2cm] {$(11)^*$};
\draw[barr] (D2D1.west) to[bend right] (AD1.west);
\node[blue, font=\small, rotate=90] at (1.0,-6.8) {$\beta\times\id$};
\draw[barr] (D2D2.west) to[bend right] (AD2.west);
\node[blue, font=\small, rotate=90] at (3.7,-6.8) {$\beta\times\id$};
\draw[barr] (D2D3.west) to[bend right] (AD3.west);
\node[blue, font=\small, rotate=90] at (5.7,-6.8) {$\beta\times\id$};
\end{tikzpicture}
\caption{An object of $\kthlevel{\minion R}^{(A)}$.}
\label{fig:kthlevel}
\end{figure}

\begin{defin}
Let $\minion R$ be a minion and $k\in\mathbb N$.
We define the \emph{$k$-th level} of $\minion{R}$, denoted $\kthlevel{\minion{R}}$, as follows.
Elements of $\kthlevel{\minion{R}}^{(A)}$  are triples $(\diag D,r,\omega)$, where 
\begin{itemize}
    \item $\diag{D}$ is a discrete LC instance with variable set $X$, \, 
    \item $r=(r_{\bar x})_{\bar{x}\in X^k}$ is a solution to $\minion R\circ \kthpower{\diag D}$, 

    \item $\omega=(\omega_{\bar{y}})_{\bar{y}\in X^{k-1}}$ is a tuple where $\omega_{\bar{y}}\in 
    \minion R^{(A\times \diag{D}_{y_1}\times\cdots\times \diag{D}_{y_{k-1}})}$
    and such that there exist
    $z\in X$ and a map $\beta\colon \diag D_z\to A$ satisfying 
    $$
    \omega_{\bar{y}} = 
    \minion R^{(\beta\times\id\times\dots\times\id)}(r_{z y_1\dots y_{k-1}})
    $$ 
    for all $\bar{y}\in X^{k-1}$.
\end{itemize}
    For a mapping $\alpha\colon A \to B$, we define  $\kthlevel{\minion R}^{(\alpha)}(\diag D,r,\omega) =(\diag D, r,\omega')$, where
    $$\omega'_{\bar{y}} :=
    \minion R^{(\alpha\times\id\times\dots\times\id)}(\omega_{\bar y}).$$ 
\end{defin}

Observe that the $k$-th level of a minion is indeed a correctly defined minion. 
With the definitions in place, the proof of our first main result becomes straightforward. It is presented in Section \ref{sec:ProofOfLevels}.

\begin{restatable}{thm}{THMLevelVsPower} \label{thm:LevelVsPower}
Suppose $\mathcal R$ is a minion, $k\in\mathbb N$, and 
$\diag D$ is a connected label cover instance. 
Then the $k$-th power of $\diag D$ has a solution in $\minion R$ 
if and only if 
$\diag D$ has a solution in $\kthlevel{\minion R}$. Schematically:
$$
\minion R\circ\kthpower{\diag D}\text{ solvable}\iff \kthlevel{\minion R}\circ \diag D\text{ solvable}
$$
\end{restatable}

We remark that the connectedness assumption in the theorem is not necessary when $\minion R$ is one of the minions describing the four basic relaxations. However, when we apply the theorem to characterize the $k$-consistency reduction (Theorem~\ref{thm:CharDatalogReductions}), we need $\minion R$ to be an arbitrary minion.

\section{Vector relaxation} \label{sec:vectors}

\subsection{Vector minion}

\renewcommand{\arraystretch}{1.3}
\setlength{\tabcolsep}{0pt}

In order to define the vector minion $\VkZp$ we first fix some notation.

We fix a prime number $p$. All computations are performed modulo $p$ -- in the field $\mathbb{Z}_p$. The weight of a vector in $\mathbb{Z}_p^N$ is denoted $\weight$ and the componentwise product of vectors by $\cm$:
\begin{align*}
\weight(v_1,v_2, \cdots, v_N) &= v_1 + v_2 + \cdots + v_N \\
(u_1, u_2, \cdots, u_N) \cm (v_1, v_2, \cdots, v_N) &= (u_1v_1, u_2v_2, \cdots, u_Nv_N) 
\end{align*}

\begin{defin} \label{def:VectorMinion}
The minion $\VkZp$ is defined as follows.
For a finite set $D$, let $\VkZp^{(D)}$ be the set of triples $(V, \vc{i}, \vc{v})$, where
$V$ is a subspace of a vector space $\mathbb Z_{p}^{N}$ (for some $N \in \mathbb{N}$), $\vc{i} \in V$, and  $\vc{v}\colon D \to V$,
satisfying the following conditions:
\begin{itemize}
    \item[(i)] $\vc{i} = (1,1,\dots,1)$ and $\weight(\vc{i})=1$.
    \item[(s)] $\sum_{a\in D} \vc{v}(a) = \vc{i}$.
    \item[(o)] $\weight(\vc{v}(a) \cm \vc{v}(b) \cm \vc{u}_1 \cm \vc{u}_2 \cm \cdots \cm \vc{u}_{k-2}) = 0$ for any $a,b \in D$ with $a \neq b$ and any $\vc{u}_1,\dots,\vc{u}_{k-2} \in V$.
    \item[(p)] $\weight(\vc{v}(a) \cm \vc{v}(a) \cm \vc{u}_1 \cm \vc{u}_2 \cm \cdots \cm \vc{u}_{k-2}) = \weight(\vc{v}(a) \cm \vc{u}_1 \cm \vc{u}_2 \cm \cdots \cm \vc{u}_{k-2})$ for any $a \in D$ and any $\vc{u}_1,\dots,\vc{u}_{k-2} \in V$.
\end{itemize}
For a function $\alpha\colon D \to E$, define $\VkZp^{(\alpha)} (V, \vc{i}, \vc{v}) = (V, \vc{i}, \summing{\alpha}{\vc{v}})$, where 
$$(\summing{\alpha}{\vc{v}})(b) = \sum_{a \in \alpha^{-1}(b)} \vc{v}(a).$$
\end{defin}

Note that, having $\mathbf i\in V$, 
conditions (o) and (p) imply 
$\weight(a\cm b) =0$ for all $a,b\in D$, $a\neq b$,
and $\weight(a\cm a) = \weight(a)$ for all $a\in D$.
Before showing that $\VkZp$  is a minion, observe that for $k=2$ the fixed component $(V,\vc{i})$ can be omitted
and condition (s) replaced by 
$\sum_{a\in D} \weight(\vc{v}(a)) = 1$. Also note that condition (p) becomes vacuous when working over $\mathbb{Z}_2$.

\begin{lem}
    $\VkZp$ is a correctly defined minion.
\end{lem}
\begin{proof}
    Since the functoriality of $\VkZp$ is clear, it suffices to observe that $(V,\vc{i},\summing{\alpha}{\vc{v}})$ satisfies the four properties. Property (i) is immediate. We verify (s) as follows:
    $$
    \sum_{b \in E} ( \summing{\alpha}{\vc{v}}  )(b) = \sum_{b \in E} \;\;\sum_{a \in \alpha^{-1}(b)} \vc{v}(a) = \sum_{a \in D} \vc{v}(a) = \vc{i}.
    $$
    To see (o), let $c,d$ be distinct elements of $E$. Since $\weight$ is linear, and the sets $\alpha^{-1}(c)$, $\alpha^{-1}(d)$ are disjoint, we obtain (o) from the same property for the original triple $(V,\vc{i},\vc{v})$, as follows:
    \begin{align*}
    \weight(\ &(\summing{\alpha}{\vc{v}} )(c) \  \cm  \ (\summing{\alpha}{\vc{v}}) (d) \ \cm \vc{u}_1 \cm \cdots \vc{u}_{k-2}) \\
    &= 
    \weight\left( \left(\sum_{a \in \alpha^{-1}(c)} \vc{v}(a) \right) 
        \cm 
        \left(\sum_{b \in \alpha^{-1}(d)} \vc{v}(b) \right) 
        \cm 
        \vc{u}_1 \cm \cdots \vc{u}_{k-2}\right) \\
    &=
    \sum_{a \in \alpha^{-1}(c)} \;\sum_{b \in \alpha^{-1}(d)} 
      \weight( \vc{v}(a) \cm \vc{v}(b) \cm \vc{u}_1 \cm \vc{u}_2 \cm \cdots \cm \vc{u}_{k-2}) 
    = 0.
    \end{align*}
    Finally, (p) follows by linearity of $\weight$ and properties (o) and (p) for the original triple:
    \begin{align*}
    \weight(\ &(\summing{\alpha}{\vc{v}} )(c) \  \cm  \ (\summing{\alpha}{\vc{v}}) (c) \ \cm \vc{u}_1 \cm \cdots \cm \vc{u}_{k-2}) \\
    &= 
    \weight\left( \left(\sum_{a \in \alpha^{-1}(c)} \vc{v}(a) \right) 
        \cm 
        \left(\sum_{a \in \alpha^{-1}(c)} \vc{v}(a) \right) 
        \cm 
        \vc{u}_1 \cm \cdots \cm \vc{u}_{k-2}\right) \\
        &= 
    \sum_{a \in \alpha^{-1}(c)}\weight(\vc{v}(a)  
        \cm 
        \vc{v}(a)
        \cm 
        \vc{u}_1 \cm \cdots \cm \vc{u}_{k-2}) \\
        &= 
    \sum_{a \in \alpha^{-1}(c)}\weight(\vc{v}(a)  
        \cm 
        \vc{u}_1 \cm \cdots \cm \vc{u}_{k-2}) \\        
        &= 
    \weight\left(\left(\sum_{a \in \alpha^{-1}(c)} \vc{v}(a) \right) 
        \cm 
        \vc{u}_1 \cm \cdots \cm \vc{u}_{k-2}\right) = \\
        &= \weight(\ (\summing{\alpha}{\vc{v}} )(c) \cm \vc{u}_1 \cm \cdots \cm \vc{u}_{k-2}).
    \end{align*}
\end{proof}

The following proposition provides the $k$-dimensional analogue of the statement from the introduction  that ``every matrix is a Gram matrix''. We need the following concept. A function $t$ from a product $M \times M \times \dots$ is called \emph{totally symmetric} if the value $t(m_1, m_2, \dots)$ depends solely on the set $\{m_1,m_2, \dots\}$.

\begin{prop} \label{prop:Gram}
Let $\spacW$ be a vector space over $\mathbb{Z}_p$, and  let
$\formG \colon \spacW \times \spacW \times \dots \times \spacW \to \mathbb{Z}_p$ be a $k$-linear form. Suppose that $M$ is a basis of $\spacW$ such that the restriction of $\formG$ to $M$ is totally symmetric.
Then, for any sufficiently large $N \in \mathbb{N}$, there exists a linear map $ \ren \colon \spacW \to \mathbb{Z}_p^N$ such that for all $w_1, w_2, \dots, w_k \in W$,
\begin{equation} \label{eq:gram}
\formG(w_1,w_2, \cdots, w_k) = \weight( \ren(w_1) \cm \ren(w_2) \cm \cdots \cm \ren(w_k)),
\end{equation}
and, moreover, 
for any $l \in \mathbb{N}$, the mapping $M \times M \times \cdots \times M \to \mathbb{Z}_p^N$ defined by $(m_1, m_2, \ldots, m_l) \mapsto \ren(m_1) \cm \ren(m_2) \cm \cdots \cm \ren(m_l)$ is totally symmetric.
\end{prop}

\begin{proof}
For the first part, it suffices to define $\ren$ on the basis $M$.
Moreover, we will construct $\ren$ so that $\ren(M)$ contains only 0/1 vectors; 
the second part then follows automatically.

Let $M_1, M_2, M_3, \dots$ be an ordering of all at most $k$-element subsets of $M$, ordered so that $M_i \not\subseteq M_j$ whenever $i<j$. We prove by induction on $i \in \mathbb{N}$ that there exists a map $\ren\colon W \to \mathbb{Z}_p^{N_i}$ satisfying Equation~(\ref{eq:gram}) for all $w_1, w_2, \dots, w_k$ with $\{w_1, \dots, w_k\}=M_j$ for some $j \leq i$. 

Assume that $\ren$ satisfies the condition for $i-1$ (for $i=1$, take $N=0$).
Increase the dimension $N$ by one, and set this new coordinate of $\ren(w)$ to 1 if $w \in M_i$, and to 0 otherwise. Notice that  $\weight(\ren(w_1) \cm \ren(w_2) \cm \cdots \cm \ren(w_k))$ does not change when $\{w_1, w_2, \dots, w_k\} \not\subseteq M_i$ and 
increases by 1 when this set is exactly  $M_i$. By repeating this construction as many times as needed, we obtain the desired $\ren$. 

Finally, we can make $N$ larger by adding zeros to all vectors.
\end{proof}

We are now ready to prove the equivalence between the $k$-th level $\mathbb{Z}_p$ relaxation and the vector relaxation defined by the minion $\VkZp$. For simplicity, we state the theorem only for connected instances. This restriction is not essential, but removing it would require some additional technical work.
\begin{thm} \label{thm:VkZp}
Let $\diag D$ be a connected LC instance. Then $\diag D$ has a solution in $\VkZp$ if and only if the $k$-th saturated power of $\diag D$ has a solution in  $\mathcal{Z}_{p}$.
\end{thm}
\begin{proof}
Fix a connected LC instance $\diag D$ with variable set $X$. Let $P$ be its set of points and, for a variable $x \in X$, let $P_x$ be the corresponding subset of $P$.
$$
P = \{ xa \ \mid x \in X, \ a \in \diag{D}_x \}, \quad
P_x = \{xa \ \mid a \in \diag{D}_x \}
$$

Before starting the proof, we introduce some notation.
Similarly as in Figure~\ref{fig:SecondLevelSolution}, we regard solutions to $\kthpower{\diag D}$ in $\VkZp$ as  $k$-dimensional arrays indexed by points:
$$
\mapf\colon P \times P \times \cdots \times P \to \mathbb{Z}_p
$$
That is, for an assignment $(f_{\bar{x}})_{\bar{x} \in X}$ for $\kthpower{\diag D}$ in $\VkZp$, we set 
$$
\mapf(x_1a_1, x_2a_2, \dots, x_ka_k) := f_{x_1x_2 \cdots x_k}(a_1a_2 \cdots a_k), $$
and vice versa.

Similarly, instead of collection of functions $(\vc{v}_x\colon \diag{D}_x \to \mathbb{Z}_p^N)_{x \in X}$, we will rather work with $\mapv\colon P \to \mathbb{Z}_p^N$ defined as 
$$
\mapv(xa) = \vc{v}_x(a).
$$

($\Rightarrow$)
Let $(V_x, \vc{i}_x, \vc{v}_x)_{x \in X}$ be a solution to $\diag D$ in $\VkZp$. Every constraint $y = \alpha(x)$ in $\diag D$ forces $V_x=V_y$ and $\vc{i}_x=\vc{i}_y$. Since $\diag D$ is connected, we thus have a single vector space $V \leq \mathbb{Z}_p^N$ and a single $\vc{i} \in V$ such that the solution is $(V,\vc{i},\vc{v}_x)_{x \in X}$.

We propose a solution to $\kthpower{\diag D}$ in $\mathcal{Z}_p$ as follows. 
$$
\mapf(p_1, \dots, p_k) = \weight(\mapv(p_1) \cm \mapv(p_2) \cm \cdots \cm \mapv(p_k))
$$
We first observe that for every $\bar{x} = x_1x_2\cdots x_k$, $f_{\bar{x}}$ is in indeed a member of $\mathcal{Z}_p^{(\kthpower{\diag D}_{\bar{x}})}$. 
\begin{equation} \label{eq:unit}
\begin{aligned}
\sum_{\bar{a} \in \diag{D}_{x_1} \times \diag{D}_{x_2} \times \cdots \diag{D}_{x_k}} f_{\bar{x}}(\bar{a}) 
&= 
\sum_{p_1 \in P_{x_1}}  \sum_{p_2 \in P_{x_1}} \cdots  \sum_{p_k \in P_{x_1}} \mapf(p_1, p_2, \ldots, p_k) \\
&=
\sum_{p_1 \in P_{x_1}}  \sum_{p_2 \in P_{x_1}} \cdots  \sum_{p_k \in P_{x_1}} \weight(\mapv(p_1) \cm \mapv(p_2) \cm \cdots \cm \mapv(p_k)) \\
&=
\weight\left( \sum_{p_1 \in P_{x_1}} \mapv(p_1) \cm \sum_{p_2 \in P_{x_1}} \mapv(p_2) \cm \cdots \cm \sum_{p_k \in P_{x_1}} \mapv(p_k)\right) \\
&=\weight(\vc{i} \cm \vc{i} \cm \cdots \cm \vc{i}) = \weight(\vc{i})=1
\end{aligned}
\end{equation}
It remains to verify that $\mapf$ satisfies the constraints; recall that it is enough to consider type (2) constraints and type (1) constraints of a special form $\alpha \times \id \times \cdots \times \id$. 

A constraint $x_1x_2\ldots x_k = \alpha \times \id \times \cdots \times \id (y_1y_2\ldots y_k)$, where $y_1 = \alpha(x_1)$ is a constraint in $\diag D$, is satisfied if and only if for all $p_i \in P_{x_i}$
\begin{equation} \label{eq:aidid}
\mapf(p_1, p_2, p_3, \dots, p_k) = \sum_{p' \in \mapalpha^{-1}(p_1)} \mapf(p',p_2,p_3, \dots, p_k).
\end{equation}
This is equivalent to the following.
$$
\weight(\mapv(p_1) \cm \mapv(p_2) \cm \cdots \cm \mapv(p_k)) =
\weight\left(\sum_{p' \in \mapalpha^{-1}(p_1)} \mapv(p') \cm \mapv(p_2) \cm \cdots \cm \mapv(p_k)\right)
$$
But we know that $\vc{v}$ satisfies the constraint $y_1 = \alpha(x_1)$, equivalently, 
$$
\forall p_1 \in P_{x_1}  
\quad 
\mapv(p_1) = \sum_{p' \in \mapalpha^{-1}(p_1)} \mapv(p'),
$$ 
and we are done in this case.

Finally, observe that the satisfaction of all constraints of type (2) in Definition~\ref{def:Spower} is equivalent to the conjunction of two claims: $\mapf$ is totally symmetric and $\mapf(p_1,p_2, \ldots, p_k)=0$ whenever two of the $p_i$ are different but belong to a single $P_x$.
The latter holds since $\vc{v}$ satisfies (o). For the former, consider $p_1, p_2, \dots, p_k$ such that $\{p_1, p_2, \dots, p_k\} = \{s_1, s_2, \dots, s_l\}$. Then, by repeated application of (i) and (p), we get
$$
\mapf(p_1, p_2, \ldots, p_k) 
= \weight(\vc{v}(p_1) \cm \vc{v}(p_2) \cm \cdots \cm \vc{v}(p_k))
= \weight(\vc{v}(s_1) \cm  \vc{v}(s_2) \cm \vc{v}(s_l) \cm \vc{i} \cm \vc{i} \cm \cdots \cm \vc{i}),
$$
from which the total symmetry follows.

($\Leftarrow$).
Let $\mapf$ be a solution to $\kthpower{\diag D}$ in $\mathcal{Z}_p$.

The rough idea for obtaining a solution to $\diag D$ in $\VkZp$ is as follows.
We apply Proposition~\ref{prop:Gram} to $\mapf$ regarded as a $k$-linear form, then define $V_x=V$ as the image of $\ren$ and $\mapv$ by $\mapv(p) := \ren(p)$  (also choosing $\vc{i} \in V$ naturally). Such an assignment would indeed satisfy some of the requirements, but not all -- in particular the constraints in $\mathcal{Z}_p \circ \diag D$ -- so we need to be a bit more careful. 

Let $\barP$ be a vector space with basis $P$ and extend $\mapf$ to a $k$-linear form $\formF$.
$$
\formF\colon \barP \times \barP \times \cdots \times \barP \to \mathbb{Z}_p, \quad
\formF = \mapf \text{ on $P^k \subseteq \barP^k$}
$$
Let $\spacQ \leq \barP$ be the radical of $\formF$, and let $\spacW$ and $\formG$ be the corresponding quotient space and quotient $k$-linear form, respectively:
$$
\spacQ = \{q \in \barP \mid \formF(q, \barP, \barP, \dots, \barP) = 0\}, \quad 
\spacW = \barP/\spacQ, \quad
\formG = \formF/\spacQ\colon \spacW \times \spacW \times \cdots \times \spacW \to \mathbb{Z}_p
$$

We now pause our construction to make several observations.
Since $\mapf$ satisfies constraints of the form $\bar{x} = \alpha(\bar y)$,
for every constraint $x_1 = \alpha(y_1)$, every point $p \in P_{y_1}$, and any points $p_2, p_3, \dots, p_k \in P$ we have~(\ref{eq:aidid}):
$$
\mapf(p, p_2, p_3, \dots, p_k) = \sum_{p' \in \mapalpha^{-1}(p)} \mapf(p',p_2,p_3, \dots, p_k).
$$
As $P$ generates $\barP$ we get that $p$ is equal to the sum of $p'$ modulo $Q$:
\begin{equation} \label{eq:AAcons}
p/Q = \left( \sum_{p' \in \mapalpha^{-1}(p)} p' \right) /Q
\end{equation}
In particular, by summing up over $p \in P_y$, we get that $\sum_{p \in P_y} p$ and
$\sum_{p' \in P_x} p'$ are equal modulo $Q$. Since $\diag D$ is a connected instance, we see that such sums are independent (still modulo $Q$) of the variable and we denote this unique element of $\spacW$ by $m$. 
\begin{equation} \label{eq:AAsum}
m = \left( \sum_{p \in P_x} p \right)/Q \quad \text{for every $x \in X$}
\end{equation}
  Using this equation, we also obtain $\formG(m,m, \dots,m) = 1$ by following the computations in Equation~(\ref{eq:unit}) backward and using $\sum_{\bar{a} \in \diag{D}_{xx \cdots x}} f_{xx \cdots x}(\bar{a}) = 1$.

Finally, since $\mapf$ satisfies the constraints $\bar{x} = \sigma^*(\bar y)$, we know that $\mapf$ is totally symmetric on $P$ and  has the ``orthogonality property'': for any $p_1, p_2, \dots, p_k$, we have $\mapf(p_1,p_2,\dots,p_k)=0$ whenever two of the $p_i$ are in the same $P_x$. Since $\formF$ agrees with $\mapf$ on $P$, it has the same properties, and so does $\formG$ modulo $Q$. 

The total symmetry of $\formG$ on $P/Q$ can be improved. Consider the expression
$$
\formG(m, m, \dots, m, p_1/Q, p_2/Q, \dots, p_l/Q),
$$
with $m$ appearing $k-l$ times.
Take $x$ so that $p_1 \in P_x$ and write $m$ as $m=\sum_{p \in P_x} p/Q$. Using $k$-linearity and orthogonality, we get
\begin{equation} \label{eq:replacem}
 \formG(m, m, \dots, m, p_1/Q, p_2/Q, \dots, p_l/Q)
=\formG(m, \dots, m, p_1/Q, p_1/Q, p_2/Q, \dots, p_l/Q).
\end{equation}
It follows that $\formG$ is totally symmetric on $P/Q \cup \{m\}$.

Returning to the construction, we  extend $m$ to a basis $M \subseteq P/Q \cup \{m\}$. This is possible as already $P/Q$ is a generating set of $\spacW$. 
We now apply Proposition~\ref{prop:Gram} 
to the subspace of $\spacW$ generated by $M\setminus \{m\}$ 
and obtain a linear map $\ren$ for some $N=1\pmod p$ with the properties stated there, including the additional one. 
We additionally put $\ren(m) = (1,1,\dots,1)$ 
and extend the linear map  $\ren$ to $\spacW$.
Equation (\ref{eq:replacem}) and $\formG(m,m, \dots,m) = 1$ 
imply that condition \ref{eq:gram} in Proposition~\ref{prop:Gram} holds 
for the extended linear map $\ren$.

Finally, we define an aspiring solution to $\diag D$ in $\VkZp$: $V$ is the image of $\ren$, $\mathbf{i} = \ren(m) = (1,\dots,1)$ and $\mapv(p) = \ren(p/\spacQ)$.
$$
V = \ren(W), \quad 
\vc{i} = \ren(m), \quad 
\mapv(p) = \ren(p/\spacQ)
$$

We need to verify the properties (i), (s), (o), and (p).
Property (i) follows from 
$\mathbf i = (1,1,\dots,1)$.
Property (s) follows from Equation (\ref{eq:AAsum}).

We verify properties (o) and (p) on generators $\vc{u}_i \in \ren(P/Q)$. We have
\begin{align*}
\weight(&\mapv(p) \cm \mapv(p') \cm \mapv(p_1 )\cm \mapv(p_2) \cm \cdots \cm \mapv(p_{k-2})) \\
&= \weight(\ren(p/Q), \ren(p'/Q), \ren(p_1/Q), \ren(p_2/Q), \dots, \ren(p_{k-2}/Q)) \\
&= \formG(p/Q, p'/Q, p_1/Q, p_2/Q, \dots, p_{k-2}/Q)
\end{align*}
If $p \neq p'$, then we get 0 by the already established property of $\formG$. If $p=p'$, then we can continue as follows:
\begin{align*}
&= \formG(m/Q, p/Q, p_1/Q, p_2/Q, \dots, p_{k-2}/Q) \\
&= \weight(\vc{i} \cm \mapv(p) \cm \mapv(p_1 )\cm \mapv(p_2) \cm \cdots \cm \mapv(p_{k-2})) \\
&= \weight(\mapv(p) \cm \mapv(p_1 )\cm \mapv(p_2) \cm \cdots \cm \mapv(p_{k-2})) 
\end{align*}

The compatibility of $\mapv$ with constraints follows from the definitions and \ref{eq:AAcons}. This concludes the proof.
\end{proof}

\subsection{Minion homomorphisms between vector minions and $\mathcal Z_{m}$}

In this subsection, we study minion homomorphisms from the vector minions $\VkZp$ to the minions $\mathcal{Z}_m$. 

For clarity of notation, we make a few convenient simplifications.
Throughout, we  work only with members of $\cdot^{(D)}$ where $D = [n] = \{1,2, \dots, n\}$. We write $\mathcal{Z}_m^{(n)}$ and $\VkZp^{(n)}$ instead of $\mathcal{Z}_m^{([n])}$ and $\VkZp^{([n])}$.  Elements of $\mathcal{Z}_m^{(n)}$ are written as tuples $(a_1,a_2, \dots, a_n) \in \mathbb{Z}_m^n$;  with membership in $\mathcal{Z}_m$ equivalent to  $a_1+a_2+ \dots +a_n=1$ in $\mathbb{Z}_m$. 
Members of $\VkZp$ are written as $(V,\vc{i},\vc{v}_1, \vc{v}_2, \dots, \vc{v}_n)$, where $\vc{i},\vc{v}_1, \vc{v}_2, \dots, \vc{v}_n \in \mathbb{Z}_p^N$. While the subscripts on the vectors $\vc{v}$ previously referred to LC variables, no LC variables appear in this subsection, so this notation should cause no confusion. 

Because several moduli arise in the calculations, we adopt the following convention:  the symbols $+, \cdot$ are implicitly understood in $\mathbb{Z}$; addition modulo $m$ is denoted by $+_m$; and $\weight_m$ denotes the integer sum of all coordinates modulo $m$.

The next lemma is a simple but crucial observation for the positive result in this section.

\begin{lem}\label{claim:secondDigitOnZp2}
There exists a polynomial $q(x,y)$ of degree at most $p$ over $\mathbb{Z}_p$ that is divisible by $xy$, and such that for all $a,b \in \{0,1, \dots, p-1\}$, 
$$
q(a,b) = \lfloor \frac{a+b}{p} \rfloor.
$$
\end{lem}
\begin{proof}
It is straightforward to check that 
$\lfloor \frac{x+y}{p}\rfloor = \binom{x+y}{p}\pmod p$.
It remains to use Vandermonde's identity
$\binom{x+y}{p} = \sum_{j=0}^{p}\binom{x}{j}\binom y{p-j}\pmod p$
and notice that
$\binom{x}{p}=\binom{y}{p}=0$ for $x,y<p$.
\end{proof}

The following proposition, in the case $p=2$, is a consequence of the “solving the square’’ result from Section~\ref{sec:SolvingSquare}. It implies that systems of linear equations over $\mathbb{Z}_4$ can be solved by the second level of the $\mathbb{Z}_2$ relaxation. This was surprising to us, and even more so after discovering a remarkably simple rounding procedure: one merely counts the number of ones modulo $4$. We are grateful to Demian Banakh
for providing the first version of the proof for  arbitrarily large primes $p$.

\begin{prop} \label{prop:VkZpToZpSquared}
    There exists a minion homomorphism 
     $\VkZpParam{p}{p} \to \cat Z_{p^2}$.
\end{prop}

\begin{proof}
We define a minion homomorphism
 $\xi\colon \VkZpParam{p}{p} \to \cat Z_{p^{2}}$ by
$$
\xi^{(n)}(V,\vc{i},\vc{v}_1, \vc{v}_2, \cdots, \vc{v}_n) = 
\left(
\frac{\weight_{p^2}(\vc{v}_1)}{\weight_{p^2}(\vc{i})},
\frac{\weight_{p^2}(\vc{v}_2)}{\weight_{p^2}(\vc{i})},
\dots
\frac{\weight_{p^2}(\vc{v}_n)}{\weight_{p^2}(\vc{i})}
\right),
$$
where the division is in $\mathbb Z_{p^2}$. It is well-defined as 
$\weight_p(\vc{i}) = 1$.

We need to prove that $\xi$ commutes with minor maps. It is enough to consider, for each $n$, the function $\alpha\colon[n] \to [n-1]$ defined by $\alpha=(1,1,2,3, \dots, n-1)$, since every function between finite sets can be written as a composition of such $\alpha$ and injective functions, for which the commutation property is straightforward. 
For such $\alpha$ we have
\begin{align*}
\xi^{(n-1)} \circ \VkZp^{(\alpha)}(V,\vc{i},\vc{v}_1. \vc{v}_2. \dots, \vc{v}_n)
&=
\left(
\weight_{p^2}(\vc{v}_1 +_p \vc{v}_2),
\weight_{p^2}(\vc{v}_3),
\dots
\weight_{p^2}(\vc{v}_{n-1})
\right)/ \weight_{p^2}(\vc{i})
\quad \text{ and } 
\\
\mathcal{Z}_{p^2} \circ \xi^{(n)}(V,\vc{i},\vc{v}_1. \vc{v}_2. \dots, \vc{v}_n)
&=
\left(
\weight_{p^2}(\vc{v}_1)+_{p^2}
\weight_{p^2}(\vc{v}_2),
\weight_{p^2}(\vc{v}_3),
\dots
\weight_{p^2}(\vc{v}_n)
\right)/\weight_{p^2}(\vc{i}),
\end{align*}
so it is enough to verify that 
$\weight_{p^2}(\mathbf v_1 +_p \mathbf v_2) = \weight_{p^2}(\mathbf v_1) +_{p^2} 
\weight_{p^2}(\mathbf v_2)
$.

We take the polynomial $q$ from Lemma \ref{claim:secondDigitOnZp2}.
By interpreting multiplication and addition in the polynomial as $\cm$ and $+$ in $\mathbb{Z}_p^N$, we can also apply it to vectors. Using the properties in Definition~\ref{def:VectorMinion}, we have
\begin{align*}
    \weight_{p^2}(\mathbf v_1+_p\mathbf v_2) 
    &= 
    \weight_{p^2}(\mathbf v_1 + \mathbf v_2 - p\cdot q(\mathbf v_1,\mathbf v_2))   \\ 
    &= 
    \weight_{p^2}(\mathbf v_1) +_{p^2}
    \weight_{p^2}(\mathbf v_2) -_{p^2}
    p \cdot \weight(q(\mathbf v_1,\mathbf v_2)) \\
    &=\weight_{p^2}(\mathbf v_1) +_{p^2}
    \weight_{p^2}(\mathbf v_2).
\end{align*}
\end{proof}

\begin{prop} \label{prop:NotVtwoZtwoToZthree}
    There is no homomorphism $\mathcal V_{2}$-$\mathbb Z_{2}\to \cat Z_{8}$.
\end{prop}
\begin{proof}
Assume that such a homomorphism $\xi$ exists. Let $V = \mathbb{Z}_2^3$.
We define the members $\epsilon_0 \in \VkZpParam22^{(3)}$, $\epsilon_1,\epsilon_2,\epsilon_3 \in \VkZpParam22^{(4)}$ by
\begin{align*}
\epsilon_0 & =(V, \vc{i}, 100, 010, 001), \\
\epsilon_1 & =(V, \vc{i}, 100, 011, 011, 011), \\
\epsilon_2 &=(V, \vc{i}, 010, 101, 101, 101),  \\
\epsilon_3 &=(V, \vc{i}, 001, 110, 110, 110). 
\end{align*}
Denote the $j$-th coordinate of 
$\xi(\epsilon_{i})$ by $s_{i,j}\in\mathbb Z_{8}$.
First, notice that $\epsilon_{i}^{(1234)} = \epsilon_{i}^{(1342)}$ for each $i\in\{1,2,3\}$. 
Hence, $(\xi(\epsilon_i))^{(1234)} = \xi(\epsilon_i^{(1234)})=
\xi(\epsilon_i^{(1342)})=
(\xi(\epsilon_{i}))^{(1342)}$ as minor maps commute with $\xi$. 
Therefore, 
$(s_{i,1},s_{i,2},s_{i,3},s_{i,4}) = (s_{i,1},s_{i,3},s_{i,4},s_{i,2})$
and 
$s_{i,2} = s_{i,3} =s_{i,4}$.
Since $\epsilon_0^{(122)} = \epsilon_1^{(1233)}$, 
we obtain 
$s_{0,1} = s_{1,1}$, $s_{0,2}+s_{0,3}=s_{1,2}$, $s_{1,3}+s_{1,4} = 0 \mod 8$.
Similarly, considering the equations
$\epsilon_0^{(212)} = \epsilon_2^{(1233)}$
and 
$\epsilon_0^{(221)} = \epsilon_3^{(1233)}$
we derive 
$s_{i,3}+s_{i,4} = 0$ 
and 
$s_{0,i} = s_{i,1}$
for any $i\in\{1,2,3\}$.
Thus, $s_{i,2} = s_{i,3} =s_{i,4}\in\{0,4\}$.
Since 
$s_{i,1}+s_{i,2}+s_{i,3}+s_{i,4}=1$ for $i\in\{1,2,3\}$, 
we derive 
$s_{i,1}\in\{1,5\}$, and therefore $s_{0,i}\in\{1,5\}$.
This contradicts $s_{0,1}+s_{0,2}+s_{0,3}=1$.
\end{proof}

\begin{prop} \label{prop:NoZtwopToZpSquared}
    There is no homomorphism $\mathcal V_{2}$-$\mathbb Z_{p}\to \cat Z_{p^2}$
    for any prime $p>2$.
\end{prop}
\begin{proof}
For $S\subseteq [p+2]$, let $\mathbf e_{S}'$ and $\mathbf f_{S}'$ be the $(p+2)$-vector over $\mathbb Z_{p}$ with 1s exactly at coordinates 
$S$ and $[p+2]\setminus S$, respectively, 
and 0s at all remaining coordinates.
Let $\mathbf e_{S}$ and $\mathbf f_{S}$ be the concatenations 
of $\frac{p+1}{2}$ copies of $\mathbf e_S'$ and $\mathbf f_S'$, respectively;
the repetition ensures 
$\sum_{i\in[p+2]}\weight(\mathbf e_{\{i\}}) =\sum_{i\in[p+2]}\weight(\mathbf f_{\{i\}}) = 1$.
To shorten, we write $\mathbf e_{i}$ and $\mathbf f_{i,j}$ instead of 
$\mathbf e_{\{i\}}$ and $\mathbf f_{\{i,j\}}$.

Set $\mu = (\mathbf e_1,\dots,\mathbf e_{p+2})$, 
$\eta = (\mathbf f_1,\dots,\mathbf f_{p+2})$, 
$\epsilon_{i,j} = (\underbrace{\mathbf f_{i,j},\dots,\mathbf f_{i,j}}_{p+1},\mathbf e_{i,j})$.
Assume that there exists a minion homomorphism $\xi\colon\mathcal V_{2}$-$\mathbb Z_{p}\to \cat Z_{p^2}$.
Let 
$\xi(\mu) = (a_1,\dots,a_{p+2})$,
$\xi(\eta) = (b_1,\dots,b_{p+2})$.
Since permutations of the first $p+1$ coordinates in $\epsilon_{i,j}$ does not change it, 
it follows that
$\xi(\epsilon_{i,j}) = (c_{i,j},\dots,c_{i,j},d_{i,j})$.
Since $\mathbf f_{i}+\mathbf f_{j} = \mathbf e_{i,j}+2\mathbf f_{i,j}$, 
we have $b_{i}+ b_{j} = 2 c_{i,j} + d_{i,j}$.
Since $p\cdot \mathbf f_{i,j} = \mathbf 0$, 
we have $p\cdot c_{i,j} = 0$.
Then $(p+1)c_{i,j} + d_{i,j} = 1$ implies 
$c_{i,j} + d_{i,j} = 1$.
Since 
$\mathbf e_{i}+\mathbf e_{j} = \mathbf e_{i,j}$, 
we have 
$a_{i} + a_{j} = d_{i,j}$.
Thus, we have 
$b_{i}+b_{j} = 2 c_{i,j} + d_{i,j} = 2-d_{i,j}=2-a_{i}-a_{j}$ for all $i,j$.
Hence, $(a_{i}+b_{i})+(a_{j}+b_{j}) = 2$ for all $i,j$. 
Therefore, $a_{i}+b_{i} = 1$ for every $i$.
Then $\sum_{i}(b_{i}+a_{i}) = (p+2)\neq 2 = \sum_{i}b_{i}+\sum_{i}a_{i}$. Contradiction.
\end{proof}

\section{Proof of Theorem~\ref{thm:LevelVsPower}} \label{sec:ProofOfLevels}

We restate the theorem of convenience.

\THMLevelVsPower*

\begin{proof}

Throughout the proof, we use the following notation. Tuples $x_1\dots x_{k-1}$ are denoted by $\bar{x}$ and single elements by $x$. The concatenation $xx_1\dots x_{k-1} $ is shortened to $x\bar{x}$ and the product of maps $\alpha_{x_1}\times\dots\times \alpha_{x_{k-1}}$ to $\alpha_{\bar{x}}$.
By $\bar{\id}$ we denote a product of identity maps $\id\times\dots\times\id$.

$(\Rightarrow)$
Let $X$ be the set of variables of $\diag D$ and let $r$ be a solution of $\minion R\circ \kthpower{D}$. 
By $\diag D'$ we denote the discrete part of $\diag D$, i.e. $\diag D$ without constraints.
We claim that the tuple 
$(\diag D', r,\omega_x)_{x\in X}$, where 
$(\omega_x)_{\bar{x}} := 
    r_{x\bar{x}}$ for $\bar{x}\in X^{k-1}$,
is a solution of $\kthlevel{\minion R}\circ \diag D$. Indeed, take any constraint $x'=\alpha(x)$ in $\diag D$, then, since $r$ is a solution of $\minion R\circ \kthpower{\diag D}$, we have
$$
\minion R^{(\alpha\times\bar{\id})}((\omega_x)_{\bar{x}})
= \minion R^{(\alpha\times\bar{\id})}(r_{x\bar{x}}) = r_{x'\bar{x}} =(\omega_{x'})_{\bar{x}},
$$
which implies
$
(\kthlevel{\minion R})^{(\alpha)}(\diag D,r,\omega_x) 
=(\diag D,r,\omega_{x'}).
$

$(\Leftarrow)$ Let $(\diag E,s,\omega_x)_{x\in X}$ be a solution of $\kthlevel{\minion R}\circ \diag D$ and let $Y$ be the set of variables of $\diag E$. 
Since the instance $\diag D$ is connected, $\diag E$ and $s$ do not depend on $x$.
For every $x\in X$ choose some $\phi(x)\in Y$ and a map $\beta_x\colon \diag E_{\phi(x)}\to \diag D_x$ such that 
    $$
    (\omega_x)_{\bar{y}} = 
    \minion R^{(\beta_x\times\bar{\id})}(s_{\phi(x)\bar{y}})
    $$
holds for all tuples $\bar{y}\in Y^{k-1}$.
Let us define an solution $r=(r_{\bar{x}})_{\bar{x}\in X^k}$ of $\minion R\circ\kthpower{\minion D}$ as follows.
$$
    r_{\bar{x}} :=
    \minion R^{(\beta_{\bar{x}})}
    (s_{\phi(\bar{x})})
    $$ 
In the verification that $r$ is indeed a solution of $\minion R\circ \kthpower{\diag D}$ we will repeatedly use the following equation.
\begin{align}
\label{eq:ClaimInProofOfLevel=Power}
\minion R^{(\id\times\beta_{\bar{x}})}  \big(\:
(\omega_{x})_{\phi(\bar{x})} \: \big)
&=
\minion R^{(\id\times \beta_{\bar{x}})}  \big(\:
\minion R^{(\beta_{x}\times \bar{\id})}
(s_{\phi(x)\phi(\bar{x})})  \:\big) \\
&=
\minion R^{(\beta_{x}\times\beta_{\bar{x}})} 
(s_{\phi(x)\phi(\bar{x})})   \nonumber  \\
&=
r_{x\bar{x}} \nonumber 
\end{align}
Now we verify that the two kinds of constraints in $\kthpower{\diag D}$ are satisfied.
\begin{itemize}
    \item[(1)] Let $\bar{x}\in X^{k-1}$ and $\alpha\colon \diag D_x\to \diag D_{x'}$ be a constraint in $\diag D$. Then
\begin{align*}
\minion R^{(\alpha\times\bar{\id})}
(r_{x\bar{x}}) 
&=
\minion R^{(\alpha\times\bar{\id})} \big( \:
\minion R^{(\id\times\beta_{\bar{x}})}  \big(\:
(\omega_x)_{\phi(\bar{x})} \: \big) \:\big)  \\
&=
\minion R^{(\id\times\beta_{\bar{x}})}  \big(\:
\minion R^{(\alpha\times \bar{\id})} \big( \:
 (\omega_x)_{\phi(\bar{x})} \: \big) \:\big)  \\
&=
\minion R^{(\id\times \beta_{\bar{x}})}\big( \:
(\omega_{x'})_{\phi(\bar{x})} \:\big)\\
&=
r_{x'\bar{x}}.
\end{align*}
The first and last equalities use Equation~(\ref{eq:ClaimInProofOfLevel=Power}), while the third is due to $(\diag E,s,\omega_x)$ being a solution.
    \item[(2)] let $\bar{x}\in X^k$ 
    and $\sigma\colon[k]\to[k]$. Then we have
\begin{align*}
\diag R^{(\sigma^*)}(r_{\bar{x}}) 
&= 
\diag R^{(\sigma^*)} \big(
\diag R^{(\beta_{\bar{x}})} (
s_{\phi(\bar{x})}
)\big)
=
\diag R^{(\sigma^*\circ \beta_{\bar{x}})}
(s_{\phi(\bar{x})}) \\
&=
\diag R^{(\beta_{\bar{x}\circ \sigma} \circ \sigma^*)}
(s_{\phi(\bar{x})})
=
\diag R^{(\beta_{\bar{x}\circ \sigma})} \big(
\diag R^{(\sigma^*)}(
s_{\phi(\bar{x})} 
)\big)\\
&=
\diag R^{(\beta_{\bar{x}\circ \sigma})}(
s_{\phi(\bar{x}\circ\sigma)} )
=
r_{\bar{x}\circ\sigma}
\end{align*}
where $\bar{x}\circ \sigma$ denotes the tuple $x_{\sigma(1)}\dots x_{\sigma(k)}$.
The first and the last equalities are again due to \ref{eq:ClaimInProofOfLevel=Power}, the fifth holds because $s$ is a solution of $\minion R\circ \kthpower{\diag E}$. The second and fourth use functoriality of $\diag R$ while for the third, one checks that $\sigma^*\circ\beta_{\bar{x}} =\beta_{\bar{x}\circ\sigma}\circ \sigma^*$. Indeed, unfolding the notation, both functions take tuples $e_1\dots e_k\in \diag E_{\phi(x_1)}\times\dots\times \diag E_{\phi(x_k)}$ as an input, apply the corresponding $\beta$ to each coordinate and then permute the coordinates according to $\sigma$.
\begin{align*}
    \sigma^*(\beta_{\bar{x}}( e_1\dots e_k )) &=
    \sigma^*\big( \beta_{x_1}(e_1)\dots \beta_{x_k}(e_k) \big)
    \\&=
    \beta_{x_{\sigma(1)}}(e_{\sigma(1)})\dots\beta_{x_{\sigma(k)}}(e_{\sigma(k)})
\\&= \beta_{\bar{x}\circ\sigma}(e_{\sigma(1)}\dots e_{\sigma(k)})
\\&= \beta_{\bar{x}\circ\sigma}(\sigma^*(e_1\dots e_k))
\end{align*}

\end{itemize}
\end{proof}

\section{Promise CSP} \label{sec:PCSP}
\renewcommand{\lc}{\mathrm{lc}_{\str A}}
\subsection{Preliminaries}

CSP instances and the corresponding LC instances are defined as in the introduction.

\begin{defin}
   A \emph{CSP instance} $\inst{I}$ consists of a set of variables $X$, a finite domain $\inst{I}_x$ for each $x \in X$, and a list of constraints of the form $x_1x_2\dots x_n \in R$, where $R \subseteq \inst{I}_{x_1} \times \inst{I}_{x_2} \times \cdots \times \inst{I}_{x_n}$. A \emph{solution} to $\inst{I}$ is a tuple $(d_x)_{x \in X}$ satisfying $d_{x_1}d_{x_2} \ldots d_{x_n} \in R$ for every constraint $x_1x_2\dots x_n \in R$.

   Given a CSP instance $\inst{I}$, we define an LC instance $\diag{D} = \lc{\inst{I}}$  as follows. The variables of $\diag{D}$ are the variables of $\inst{I}$, with $\diag D_x = \inst{I}_x$, together with one additional variable for each constraint of $\inst{I}$. For a constraint $x_1x_2\dots x_n \in R$, the corresponding new variable $y$ has domain $R$ and we impose in $\diag{D}$ the constraints $x_1 = \pi_1(y)$, $x_2 = \pi_2(y)$, \dots $x_n = \pi_n(y)$, where $\pi_i$ is the projection function $a_1a_2\dots a_n \mapsto a_i$.   
\end{defin}

Next, we define the finite-domain fixed-template PCSPs.
A \emph{signature} $\Sigma$ is a set of symbols, each $R \in \Sigma$ has an associated \emph{arity} $\ar(R) \in \mathbb{N}$. A \emph{structure} $\str{A}$ in  signature $\Sigma$ consists of a finite set $A$ and, for each symbol $R$ in the signature, a relation $R^{\str{A}} \subseteq A^{\ar(R)}$, called the \emph{interpretation} of $R$ in $\str{A}$. 
A \emph{homomorphism} $\str A \to \str B$ between structures of the same signature is a mapping $A \to B$ preserving the relations.

\begin{defin}
    A pair of structures $(\str A,\str B)$ in the same finite signatures is called a \emph{template} of a PCSP if there exists a homomorphism $\str A \to \str B$.

    Let $(\str A,\str B)$ be a PCSP template, the PCSP over $(\str A,\str B)$, denoted $\PCSP(\str A,\str B)$, is the following promise problem. An \emph{instance} $\inst{I}$ consists of a list of formal expressions of the form $x_1x_2 \dots x_n \in R$ (where $n=\ar(R)$). The task is to decide whether the instance is solvable in $\str A$ (that is, when the symbols are interpreted in $\str A$; we denote this CSP instance $\inst{I}^{\str A}$ and the corresponding label cover instance a $\lc(\inst I)$) or not even solvable in $\str B$. 
\end{defin}

\begin{defin}
    Let $\minion{R}$ be a minion and $(\str A,\str B)$ a PCSP template. We say that the relaxation $\minion{R}$ \emph{solves} $\PCSP(\str A,\str B)$ if for every instance $\inst I$ of the PCSP, if $\lc{(\inst I)}$ has a solution in $\minion{R}$, then $\inst I$ has a solution in $\str B$.
\end{defin}

\begin{defin}
    Given a PCSP template $(\str A, \str B)$, its \emph{polymorphism minion} $\Pol(\str A,\str B)=\minion M$ is defined as follows.
    For every finite set $D$ let $\minion M^{(D)}$ be the set of all homomorphism from the $D$-th cartesian power of $\str A$ to $\str B$. Such homomorphism are called polymorphisms from $\str A$ to $\str B$.
    $$
    \minion M^{(D)} = \hom(\str A^D,\str B)
    $$
    For every map $\alpha\colon D\to E$, let $\minion M^{(\alpha)}\colon\minion M^{(D)}\to \minion M^{(E)}$ be the map that takes a polymorphism $f$ and produces its $\alpha$-minor as follows.
    \begin{align*}
        \minion M^{(\alpha)}(f) \colon \str A^E\to \str B,\quad a\mapsto f( a\circ \alpha)
    \end{align*}
\end{defin}
We say that a minion $\minion M$ is \emph{locally finite} if all the sets $\minion M(D)$ are finite. For example, the minion $\mathcal{Z}_n$ and polymorphism minions $\Pol(\str A,\str B)$ are locally finite, while the minions describing AIP, BLP and SDP are not.

\begin{defin}
Given a number $S$ and a minion $\minion R$ we define $\RelaxLCn{\minion R}$ to be the following promise problem. Given a label cover instance $\diag D$, where all sets $\diag D_x$ are of size at most $S$, decide whether $\diag D$ is solvable or not even $\minion R\circ \diag D$ is solvable. 
\end{defin}

\begin{thm}[\cite{bbko2021}, Theorem~3.12]
\label{thm:fundamental}
Let $(\str A,\str B)$ be a template of a PCSP. For every sufficiently large number $S \in \mathbb{N}$, 
the problems $\PCSP(\str A,\str B)$ and $\RelaxLCn{\Pol(\str A,\str B)}$ are equivalent up to log-space reductions.
\end{thm}

The reduction from $\PCSP(\str A,\str B)$ to $\RelaxLCn{\Pol(\str A,\str B)}$ is given by 
$
\inst I\mapsto \lc(\inst I)
$.
The other reduction, from $\RelaxLCn{\Pol(\str A,\str B)}$ to $\PCSP(\str A,\str B)$, is described in detail in \cite{bbko2021} and \cite{hadekJaklOprsal2026categories}, it takes a label cover instance $\diag D$ as an input and computes an instance of $\PCSP(\str A,\str B)$ which we denote as $\csp_{\str A}{(\diag D)}$. Curiously, this construction depends only on $\str A$ and not on $\str B$. Moreover, we remark that solutions of $\diag D$ in $\Pol(\str A,\str B)$ are in one to one correspondence with solutions of $\csp_{\str A}(\diag D)$ in $\str B$.

\subsection{Comparison with the partial solution-based levels}

While there are several (mostly equivalent) ways of defining levels of relaxations, perhaps the most prominent one uses partial solutions of a given instance~\cite{berkholz2017linear,babpcp.datalog,ChanHawTOFoolHierarchies,dalmauOprsal24reductions}. In our language it can be describes as follows.

Let $\inst I$ be an instance of $\PCSP(\str A,\str B)$ then we define a label cover instance $\partialkpower{\inst I}$ as follows. Each set $U$ of at most $k$ many variables of $\inst I$ becomes a variable of $\partialkpower{\inst I}$, where its domain $\partialkpower{\inst I}_U$ is the set of partial solutions of $I$ on $U$, i.e.\! assignments $f\in A^U$ satisfying all constraints $R(x_1,\dots x_n)$ with each $x_i\in U$. Moreover, for every inclusion $U\subseteq U'$, let $\partialkpower{\inst I}$ contain the restriction map $D_{U'}\to D_U, f\mapsto f|_U$ as a constraint.

To solve a CSP instance $\inst I$ using the $k$-th level of a relaxation $\minion R$ then simply 
means to solve the label cover instance $\partialkpower{\inst I}$ in $\minion R$.
By varying $\minion R$ one obtains descriptions of several well studies hierarchies, such as $k$-consistency when $\minion R$ describes AC and Sherali-Adams when $\minion R$ describes BLP.

The saturated power (Definition~\ref{def:Spower}) provides an a priori different way to "level up" relaxations: start again with a CSP instance $\inst I$, compute the $k$-th saturated power of the label cover instance corresponding to $\inst I$ and then solve it in $\minion R$.
In this section we show that these two methods coincide up to slight changes of $k$. In particular:

\begin{thm} \label{thm:VariableLevelsVsConstraintLevels}
    Let $\cat R$ be a relaxation that does not accept all instances and let $\inst I$ be an instance of $\PCSP(\str A,\str B)$.
    \begin{enumerate}
        \item If $\partialnpower{nk}{I}$ has a solution in $\minion R$, then so does $\kthpower{\lc (\inst I)}$, where $n$ is the maximal arity of the relations in $\str A$.
        \item If $\nthpower{k+1}{\lc(\inst I)}$ has a solution in $\minion R$, then so does $\partialnpower{k}{I}$.
    \end{enumerate}
\end{thm}

The assumption  on $\minion R$ is necessary in the theorem for a simple reason: the partial solution approach rejects instances when there is no partial solution on some small set of variables, while the saturated power cannot reject on its own. To prove item $2.$, we need the concept of \emph{dummies} and a standard fact about minions.

\begin{defin}
    Let $\minion R$ be a minion, $d\in D$ and $r\in\minion R^{(D)}$. We say that $d$ is \emph{dummy} in $r$ if there is a map $\alpha\colon D'\to D$ and $r'\in\minion R^{(D')}$ such that $d$ is not in the image of $\alpha$ and $r=\minion R^{(\alpha)}(r')$. A point $r\in\minion R^{(D)}$ is called a \emph{constant}, if all $d\in D$ are dummy in $r$.
\end{defin}

It is not hard to see that a relaxation $\minion R$  contains a constant if and only if all label cover instances $\diag D$ have a solution in $\minion R$. To illustrate these notions, consider the minion $\minion R=\mathcal{Z}_p$, then $r\in\minion R^{(D)}$ is a tuple $(r_d)_{d\in D}\in \mathbb{Z}_p^D$ and $d\in D$ is dummy $r$, precisely when $r_d=0$. The condition $\sum_{d\in D} r_d =1$ guarantees that there are no constants.

\begin{lem}[Essentially in \cite{trnkova1971}]
\label{lem:dummies}
    Assume that $\minion R$ does not contain a constant and let $D'\subseteq D$ be a subset that contains all non-dummy elements of some $r\in \minion R^{(D)}$. Then there is some $r'\in\minion R^{(D')}$ such that $r=\minion R^{(\iota)}(r')$, where $\iota$ is the inclusion map $D'\to D$. 
    In particular, if $\alpha\colon D\to E$ is a map and all $\alpha$-preimages of $e\in E$ are dummy in $r$, then $e$ is dummy in $\minion R^{(\alpha)}(r)$.
\end{lem}

\begin{proof}[Proof of Theorem~\ref{thm:VariableLevelsVsConstraintLevels}.1.]
    We shorten the names of the label cover instances $\partialnpower{nk}{I}$ and $\kthpower{\lc (\inst I)}$ to $\diag P$ and $\diag D$ respectively.

    Recall that the variables of $\diag D$ are tuples $x=x_1\dots x_k$, where each $x_i$ is either a variable or a constraint of $\inst I$. For each such tuple $x$, we define $\phi(x)$ to be the set of all variables of $\inst I$ that appear either directly or in any constraint appearing in $x$. Note that since the arity of constraints of $\inst I$ is bounded by $n$, the size of $\phi(x)$ cannot exceed $nk$.

    Moreover, for each tuple $x$ let $\beta_x\colon\diag P_{\phi(x)}\to \diag D_x$ be the map that sends each partial solution $f\colon \phi(x)\to A$ to the tuple $f(x_1)\dots f(x_k)$. If $x_i$ is a constraint $y_1\dots y_\ell\in R$, then by $f(x_i)$ we mean the tuple $f(y_1)\dots f(y_\ell)$ which is in $R^\str A$ because $f$ is a partial solution.

    Let $U\mapsto r_U$ be an $\minion R$-solution of $\diag P$, then we claim that the assignment $x\mapsto\minion R^{(\beta_x)}(r_{\phi(x)})$ is an $\minion R$ solution of $\diag D$. We confirm that the two kinds of constraints in $\diag D$ are satisfied:
    \begin{enumerate}
        \item Consider a constraint $x_1x_2\ldots x_k = \alpha \times \id \times \cdots \times \id (y_1y_2\ldots y_k)$
        in $\diag D$, where $y_1 = \alpha(x_1)$ is  a constraint in $\lc (\inst I)$,
        and so $\alpha$ is a projection $R^\str A\to A$. We have an inclusion $\phi(y)\subseteq\phi(x)$, so there is a restriction map $\mathrm{res}\colon\diag P_{\phi(x)}\to\diag P_{\phi(y)}$ in $\diag P$. Moreover we claim that 
        $(\alpha\times\id\times \dots\times\id)\circ \beta_x= \beta_y\circ\mathrm{res}$. Indeed, given a partial solution $\phi(x)\to A$, it does not matter whether one first applies $f$ to every $x_i$ component-wise and then forgets some components by projection, or if one first projects and then applies $f|_{\phi(y)}$.
        The above equation also holds for the corresponding minor maps:
        \begin{align*}
            \minion R^{(\alpha\times \id\times\dots\times\id)}\big(
            \minion R^{(\beta_x)}(r_{\phi(x)})\big) =
            \minion R^{(\beta_y)}\big(
            \minion R^{(\mathrm{res})}(r_{\phi(x)})\big)=
            \minion R^{(\beta_y)}(r_{\phi(y)})
        \end{align*}

        \item Let $\sigma\colon[k]\to[k]$ be a map and $x$ a variable in $\diag D$ and $\sigma^*\colon \diag D_x\to \diag D_{x\circ\sigma}$ be the corresponding constraint. Every entry of the tuple $x\circ\sigma$ also appears in $x$, so $\phi(x\circ\sigma)\subseteq\phi(x)$, hence there is a restriction map $\mathrm{res}\colon \diag P_{\phi(x)}\to \diag P_{\phi(x\circ\sigma)}$ in $\diag P$. Similarly to above, we claim that $\sigma^*\circ\beta_x=\beta_{\sigma_x}\circ \mathrm{res}$. Indeed, it does not mater if one first applies a partial solution $f$ to all entries of a tuple and then permutes the entries with $\sigma$, or if the tuples are first permuted and $f|_{\phi(x\circ\sigma)}$ applied to them afterwards.
        \begin{align*}
            \minion R^{(\sigma^*)}\big(
            \minion R^{(\beta_x)}(r_{\phi(x)})\big) =
            \minion R^{(\beta_{x\circ\sigma})}\big(
            \minion R^{(\mathrm{res})}(r_{\phi(x)})\big)=
            \minion R^{(\beta_{x\circ\sigma})}(r_{\phi(x\circ\sigma)})
        \end{align*}
    \end{enumerate}
\end{proof}

\def\lcirc{\hspace{0.1em}\circ\hspace{0.1em}}

\begin{proof}[Proof of \ref{thm:VariableLevelsVsConstraintLevels}.2.]
    We shorten the names of $\partialnpower{k}{\inst I}$ and $\nthpower{k+1}{\inst I}$ to $\diag P$ and $\diag D$ respectively.

For every set $U$ of at most $k$ many variables of $\inst I$, pick a tuple $\phi(U)= u_1\dots u_ku_k\in U^{k+1}$ which contains all elements of $U$ and the last two entries coincide. Moreover pick for all $U$ a map $\psi(U)\colon U\to [k+1]$ which is right inverse to $\phi(U)$, i.e. $\phi(U)\circ\psi(U)= \id_U$. 

Take a solution $r=(r_x)_{x\in X^{k+1}}$ of $\minion R\circ \diag D$. For every set $U$, we define an object $s_U\in\minion R^{(A^U)}$ as 
    $$
    s_U\coloneq \minion R^{(\psi(U)^*)}(r_{\phi(U)})
    $$
where $\psi(U)^*$ is the map $A^{k+1}\to A^U, a\mapsto a\circ\psi(U)$.
We claim that each function $f\in A^U$ is a partial solution of $\inst I^\str A$ or a dummy of $s_U$. Assume that $f$ is not a partial solution and pick a constraint $c= (u_{i_1}\dots u_{i_\ell}\in R)$ in $\inst I$ with 
$f(u_{i_1})\dots f(u_{i_\ell})\notin R^\str A$. 

Step 1: let $t_1=u_1\dots c \dots u_k c$ be $\phi(U)$, but with the $i$-th and last entry replaced by $c$. Then any tuple $a_1\dots \bar{a}'\dots a_k\bar{a}\in \diag D_{t_1}$ with $\bar{a}\neq\bar{a}'$ is dummy in $r_{t_1}$. Indeed, such tuples are not in the image of the map 
$$
\sigma^*= (1\dots k+1\dots k,k+1)^*\colon \diag D_{t_1}\to\diag D_{t_1},
$$
but $r_{t_1} = \minion R^{(\sigma^*)}(r_{t_1})$, because $r$ is a solution.

Step 2: let $t_2=u_1\dots u_k c$, then any tuple $a =a_1\dots a_k\bar{a}\in \diag D_{t_2}$ with  $a_{i_j}\neq \pi_j(\bar{a})$ is dummy in $r_{t_2}$. 
Indeed, let $\alpha\colon \diag D_{t_1}\to\diag D_{t_2}$ be the map which is $\pi_j$ in the $i_j$-th coordinate and the identity everywhere else. Then all $\alpha$ preimages of $a$ must be of the form $(a_1\dots \bar{a}'\dots a_k \bar{a})$ with $\bar{a}'\neq\bar{a}$. 
Hence, by Step 1 and Lemma~\ref{lem:dummies}, $a$ must be a dummy of $r_{t_2}$. 

Step 3: Take $\phi(U)= u_1\dots u_ku_k$, then we claim that each tuple $a=a_1\dots a_{k+1}$ in $\diag D_{\phi(U)}$ with $a_{i_1}\dots a_{i_l}\notin R^\str A$ is dummy in $r_{\phi(U)}$. Indeed, consider the map 
$$
\sigma^*= (1\dots kk)^*\colon \diag D_{t_2}\to\diag D_{\phi(U)}
$$
and assume $a_{i_1}\dots a_{i_l}\notin R^\str A$ and let $a'=a_1\dots a_k\bar{a}$ be a $\sigma^*$-preimage of $a$. Since $\bar{a}\in R^\str A$, there must be a coordinate $j$ with $\pi_j(\bar{a})\neq a_{i_j}$. This means that all such $a'$ are dummies of $r_{t_2}$ by Step 2, hence $a$ is a dummy of $r_{\phi(U)}$ by Lemma~\ref{lem:dummies}.

Step 4: let $f\in A^U$. If $f$ is not of the form $a\circ\psi(U)$ for some $a\in A^{k+1}$ then it is a dummy of $s_U$. If it is, then $f(u_i) = a_{\psi(U)(u_i)}$ and hence $a_{\psi(U)(u_{i_1})}\dots a_{\psi(U)(u_{i_l})}$ is not in $R^\str A$. Hence $a$ must be a dummy of $r_{\phi(U)}$ by Step 3, so $f$ is a dummy of $s_U$.

We have shown that every $f\in A^U$ which is not a partial solution is a dummy of $s_U\in \minion R^{(A^U)}$. Let $\iota_u$ be the inclusion map $\diag P_U\subseteq A^U$, then by Lemma~\ref{lem:dummies}, there is $s'_U\in \minion R^{(\diag P_U)}$ such that $s_U=\minion R^{(\iota_U)}(s'_U)$. We claim that $U\mapsto s'_U$ is a solution of $\minion R\circ \diag P$.

Step 5: let $\iota\colon U'\subseteq U$ be an inclusion map and let $\mathrm{res}\colon \diag A^U\to\diag A^{U'}$ be the corresponding restriction. We claim that 
$
s_{U'} = R^{(\mathrm{res})}(s_U)
$.
Consider the map $\sigma=\psi(U)\circ\iota\circ\phi(U')$ from $[k+1]$ to $[k+1]$. Then $\phi(U)\circ \sigma = \iota\circ\phi(U')$, so there is a constraint $\phi(U') = \sigma^*(\phi(U))$ in $\diag D$, hence $r_{\phi(U')}=\minion R^{(\sigma^*)}(r_{\phi(U)})$. Moreover, $\psi(U')^*\circ \sigma^*=\mathrm{res}\circ\phi(U)^*$ which in turn gives the following.
\begin{align*}
    \minion R^{(\mathrm{res})}(s_U) &= 
    \minion R^{(\mathrm{res})}\big( \minion R^{(\psi(U)^*)}(r_{\phi(U)})\big) \\&= 
    \minion R^{(\psi(U')^*)}\big(\minion R^{(\sigma^*)}(r_{\phi(U)})\big) \\&=
    \minion R^{(\psi(U')^*)}(r_{\phi(U')}) \\&=
    s_{U'}
\end{align*}

Step 6: finally, let $\mathrm{res}'\colon \diag P_U\to\diag P_{U'}$ be the restriction constraint in $\diag P$ and let $p\colon A^{U'}\to\diag P_{U'}$ be left inverse to the inclusion $\iota_{U'}$.
\begin{align*}
    \minion R^{(\mathrm{res}')}(s'_U) &=
    \minion R^{(p\lcirc\iota_{U'}\lcirc\mathrm{res}')}(s'_U) \\&=
    \minion R^{(p\lcirc \mathrm{res}\lcirc\iota_{U'})}(s'_U) \\&=
    \minion R^{(p\lcirc \mathrm{res})}(s_U) \\&=
    \minion R^{(p)}(s_{U'}) \\&=
    \minion R^{(p\lcirc\iota_{U'})}(s'_{U'}) = 
    s'_{U'}
\end{align*}
\end{proof}

\subsection{Compactness}

Proposition~\ref{prop:HomosPreserveSolutions} shows that a minion homomorphism $\minion{R} \to \minion{S}$ implies that the $\minion{R}$-relaxation is stronger than the $\minion{S}$-relaxation: for any LC instance $\diag D$,  if $\diag D$ is solvable in $\minion{R}$, then it is solvable in $\minion{S}$. The following proposition states that the converse implication holds as well, provided $\minion{S}$ is locally finite. This simple compactness statement may be of independent interest.

\begin{prop} \label{prop:compatctness}
Let $\cat R$ and $\cat S$ be minions, and suppose that $\cat S$ is locally finite. Then the following are equivalent.
\begin{enumerate}
    \item There exists a minion homomorphism $\minion{R} \to \minion{S}$.
    \item For every LC instance $\diag D$, if $\minion{R} \circ \diag{D}$ is solvable, then so is $\minion{S} \circ \diag {D}$.
\end{enumerate}
\end{prop}

\begin{proof}
    Consider the finite sets $\minion S^{(D)}$ as discrete (and as they are finite, compact) topological spaces, then the product space 
    $$
    \prod_{\substack{D \text{ finite} \\ r\in\minion R^{(D)}}} \minion S^{(D)}
    $$
    is also compact by Tychonoff’s theorem. In particular, if a family of closed subsets has empty intersection, there is a finite subfamily with empty intersection.
    Let $\Phi$ be the sets of all pairs $(\alpha, r)$, where $\alpha\colon D\to E$ is a map and $r\in\minion R^{(D)}$. Then the set of minion homomorphisms from $\minion R$ to $\minion S$ can be written as the following intersection of closed subsets.
    $$
    \bigcap_{(\alpha,r)\in \Phi}
    \{ (s_r)_{r\in\minion R} \mid 
    \minion S^{(\alpha)}(s_r) = s_{\minion R^{(\alpha)}(r)} \} \subseteq 
    \prod_{\substack{D \text{ finite} \\ r\in\minion R^{(D)}}} \minion S^{(D)}
    $$
    Assuming that there is no minion homomorphism $\minion R\to \minion S$, we can find a finite subset $\Phi_0\subseteq \Phi$ such that the intersection
    $$
    \Psi\coloneq \bigcap_{(\alpha,r)\in \Phi_0}
    \{ (s_r)_{r\in\minion R} \mid 
    \minion S^{(\alpha)}(s_r) = s_{\minion R^{(\alpha)}(r)} \} 
    $$
    is empty. We now construct a label cover instance $\diag D$, such that $\minion R\circ \diag D$ is solvable but $\minion S\circ \diag D$ is not.
    Let the set of variables $X$ of $\diag D$ consist of all $r$ and all $\minion R^{(\alpha)}(r)$ with $(r,\alpha)\in \Phi_0$ and let $\diag D_r$ be the set $D$, whenever $r\in\minion R^{(D)}$. Moreover, for each pair $(r,\alpha)$ in $\Phi_0$, add the constraint $\minion R^{(\alpha)}(r) =\alpha(r)$.
    The set of solutions of $\minion S\circ \diag D$ is then precisely
    $$
    \bigcap_{(\alpha,r)\in \Phi_0}
    \{ (s_r)_{r\in X} \mid 
    \minion S^{(\alpha)}(s_r) = s_{\minion R^{(\alpha)}(r)} \} 
    $$
    which must be empty because $\Psi$ is empty.
    But $\minion R\circ D$ has a solution, namely $(r)_{r\in X}$.
\end{proof}

\subsection{Relaxations and minion homomorphisms}

Proposition~\ref{prop:compatctness} applied to a polymorphism minion $\minion{S} = \Pol(\str A,\str B)$ implies that a minion homomorphism $\minion{R} \to \Pol(\str A,\str B)$ exists whenever solvability of $\minion R\circ \diag D$ implies solvability of $\Pol(\str A,\str B ) \circ \diag{D}$. This implication is almost the same as saying that $\minion{R}$ solves $\PCSP(\str A,\str B)$, except that $\diag D$ is an arbitrary LC instance in the former and $\diag D = \lc(\inst I)$ (for a PCSP instance $\inst{I}$) in the latter. This difference is, however, immaterial by the PCSP theory~\cite{bbko2021,hadekJaklOprsal2026categories}. 

\begin{thm} \label{thm:RelaxationsAndHomos}
    Let $\minion{R}$ be a minion and $(\str A,\str B)$ a PCSP template. Then the $\minion R$-relaxation solves $\PCSP(\str A,\str B)$ if and only if  exists a minion homomorphism $\minion R\to\Pol(\str A,\str B)$.
\end{thm}

To any label cover instance $\diag D$, one can associate a minion $\gl{\diag D}$ with the property that minion homomorphisms form $\gl{\diag D}$ to any other minion $\minion R$ are in one to one correspondence with solutions of $\minion R\circ \diag D$, see Lemma~3.14 in \cite{hadekJaklOprsal2026categories}. 
$$
\hom(\gl{\diag D},\minion R) = \{\text{solutions of }\minion R\circ\diag D\}
$$
Additionally, $\gl{\diag D}$ has the property that the LC instance $\lc\csp_{\str A}(\diag D)$ has a solution in $\gl{\diag D}$ for all $\str A$. This can be observed directly from Lemma~3.11 in \cite{hadekJaklOprsal2026categories}, once the definitions of $\lc$ and $\csp_{\str A}$ are unpacked (in the above citation, they are denoted as $(\str A\circ-)$ and $\mathrm{gr}(\gl{-}\circ \str A)$ respectively).

\begin{proof}[Proof of Theorem~\ref{thm:RelaxationsAndHomos}]
    Assume that $\minion R$ solves $\PCSP(\str A,\str B)$, which means that the following implication holds for every CSP instance $\inst I$. If $\minion R\circ \lc(\inst I)$ has a solution, then $\inst I$ has a solution in $\str B$.

    By Proposition~\ref{prop:compatctness}, it suffices to show that for a given LC instance $\diag D$, if $\minion R\circ\diag D$ has a solution, then so does $\minion M\circ \diag D$, where $\minion M= \Pol(\str A,\str B)$. Indeed, if $\minion R\circ\diag D$ has a solution, then there is a minion homomorphism $\gl{\diag D}\to\minion R$, hence the LC instance $\minion R\circ\lc{\csp_{\str A}(\diag D)}$ has a solution. Therefore, by assumption, the CSP instance $\csp_{\str A}({\diag D})$ has a solution in $\str B$. But $\mathrm{csp}_{\str A}$ is a reduction from $\RelaxLC{\minion M}$ to $\PCSP(\str A,\str B)$, which implies that $\minion M\circ \diag D$ has a solution.
\end{proof}

\begin{cor}
\label{cor:kthLevelChar}
    The $k$-th level of a relaxation $\minion R$ (defined using the $k$-th saturated power) solves $\PCSP(\str A, \str B)$ if and only if there is a minion homomorphism from $\kthlevel{\minion R}$ to $\Pol(\str A,\str B)$.
\end{cor}

\begin{proof}
    We show that the $k$-th level of a relaxation $\minion R$ solves $\PCSP(\str A,\str B)$ if and only if there it is solved by the $\kthlevel{\minion R}$ relaxation. Let $\inst I$ be an instance $\diag D$ the corresponding label cover instance, then the claim follows from Theorem~\ref{thm:RelaxationsAndHomos}.

    $(\Leftarrow)$  If  $\kthpower(\diag D)$ has a solution in $\minion R$, then $\diag D$ has a solution in $\kthlevel{\minion R}$ by Theorem~\ref{thm:LevelVsPower}. We remark that this
implication does not need D to be connected. Hence, $\inst I$ has a solution in $\str B$.

    $(\Rightarrow)$ Let $\inst I'$ be a connected component of $\inst I$. If $\lc(\inst I)$ has a solution in $\kthlevel{\minion R}$, then so does $\lc(\inst I')$. Since $\lc(\inst I')$ is connected, Theorem~\ref{thm:LevelVsPower} implies that $\kthpower{\lc(\inst I')}$ has a solution in $\minion R$. But we assume that the problem is solved by $k$-consistency, hence $\inst I'$ has a solution in $\str B$. As this holds for all connected components, $\inst I$ also has a solution in $\str B$.
    Note that this implication of the theorem does not require $\diag D$ to be connected.
\end{proof}

As a particular case of the above corollary, we obtain a minion characterization of the Sherali-Adams hierarchy for PCSPs. 

\begin{cor}
    $\PCSP(\str A,\str B)$ is solved by some level of the Sherali-Adams hierarchy if and only if there is a minion homomorphism $\kthlevel{\minion R}\to \Pol(\str A,\str B)$ for some $k$, where $\minion R$ is the minion describing BLP.
\end{cor}

\subsection{$k$-consistency reduction}
\def\arc{\mathrm{arc}}

The $k$-consistency relaxation has been generalized in \cite{dalmauOprsal24reductions} to function as a reduction between two different PCSPs. Their construction can be described as follows.

Start with an instance $\inst I$ of $\PCSP(\str A,\str B)$ and compute the partial solution label cover instance $\diag D=\partialkpower{\inst I}$. Then, using arc consistency, find the largest subsets $\diag D'_x\subseteq D_x$ such that $\diag D'_y=\alpha(\diag D'_x)$ for all constraints $y=\alpha(x)$ in $\diag D$. We denote the resulting label cover instance as $\arc(\diag D)$.
To obtain an instance of a different $\PCSP(\str A',\str B')$, apply the standard reduction $\csp_{\str A'}$.

For $k=1$ there was a minion characterization provided in \cite{dalmauOprsal24reductions} using the following construction. Given a minion $\minion R$, let $\omega\minion R$ be the minion, where $\omega\minion R^{(D)}$ is the set of pairs $(D',r)$, where $D'$ is a subset of $D$ and all elements $d\in D\setminus D'$ are dummies of $r$. For maps $\alpha\colon D\to E$, we have $\omega\minion R^{(\alpha)}(D',r)= (\alpha( D'),\minion R^{(\alpha)}(r))$.
The characterization then follows directly from the following 
statement, see Lemma C.1 in the above citation:
Given a label cover instance $\diag D$ and a minion $\minion R$, $\arc(\diag D)$ has a solution in $\minion R$ if and only if $\diag D$ has a solution in $\omega\minion R$.  

Combining this result with Theorem~\ref{thm:LevelVsPower} yields the following characterization of the $k$-consistency relaxation.

\begin{thm} \label{thm:CharDatalogReductions}
    $\PCSP(\str A',\str B')$ reduces to $\PCSP(\str A,\str B)$ via the $k$-consistency reduction (defined using the $k$-th saturated power) if and only if there is a minion homomorphism from $\kthlevel{(\omega\minion M)}$ to $\minion M'$,
    where $\minion M$ and $\minion M'$ are the polymorphism minions of $(\str A,\str B)$ and $(\str A',\str B')$ respectively.
\end{thm}

\begin{proof}
$(\Leftarrow)$
We show that $\arc\circ\mathrm{pow}_k$ is a reduction from $\RelaxLC{ \minion M'}$ to $\RelaxLC{\minion M}$, then the claim follows from Theorem~\ref{thm:fundamental}.

Take a label cover instance $\diag D$ and assume that $\arc(\kthpower{\diag D})$ has a solution in $\minion M$. Then $\kthpower{\diag D}$ has a solution in $\omega\minion M$, hence $\diag D$ has a solution in $\kthlevel{(\omega\minion M)}$ by Theorem~\ref{thm:LevelVsPower}. We remark that this implication does not need $\diag D$ to be connected. Because of the assumed minion homomorphism, $\diag D'$ also has a solution in $\minion M'$. So all connected components of $\diag D$ have a solution in $\minion M'$, hence $\diag D$ has a solution in $\minion M'$ as well.

$(\Rightarrow)$ 
We show that $\PCSP(\str A',\str B')$ is solved by the $\kthlevel{\omega\minion M}$-relaxation, then claim then follows from Theorem~\ref{thm:RelaxationsAndHomos}. 

Take an instance $\inst I$ of $\PCSP(\str A',\str B')$ and let $\inst I'$ be a connected component of it. Assume that $\mathrm{lc}_{\str A'}(\inst I)$ has a solution in $\kthlevel{\omega\minion M}$, then so does $\mathrm{lc}_{\str A'}(\inst I')$. But this LC instance is connected, therefore 
$
\arc(\kthpower{\mathrm{lc}_{\str A'}(\inst I')})
$
has a solution in $\minion M$. But $\minion M$ solutions of diagrams $\diag D$ are in one to one correspondence with solutions of $\csp_{\str A}{\diag D}$ in $\str B$, hence
$
\mathrm{csp}_{\str A}(
\arc(
\kthpower{\mathrm{lc}_{\str A'}(\inst I')
}))
$
has a solution in $\str B$. But we assumed that $k$-consistency is a reduction, hence the component $\inst I'$ must have a solution in $\str B'$. As this holds for all connected components, $\inst I$ also has a solution in $\str B'$.
\end{proof}

\section*{Acknowledgments}

We are grateful to Michael Kompatscher for providing the $\mathbf{D}_4$ example, to Zarathustra Brady for showing that the vector algorithm over the integers is not yet powerful enough, to Demian Banakh for working out the first version of the $\mathbb{Z}_{p^2}$ rounding for an arbitrary prime $p$,
and to Petar Markovi\'c for fruitful discussions on vector minions.

\bibliographystyle{alpha}
\newcommand{\etalchar}[1]{$^{#1}$}

\end{document}